%% file: arxiv.tex
\documentclass[11pt]{article}
\input{header}

\newcommand{\CSP}{\ensuremath{\mathsf{CSP}}\xspace}
\newcommand{\CSPs}{\ensuremath{\mathsf{CSPs}}\xspace}
\newcommand{\SAT}{\ensuremath{\mathsf{SAT}}\xspace}
\newcommand{\MinCSP}{\ensuremath{\mathsf{Min}\text{-}\mathsf{CSP}}\xspace}

\newcommand{\MaxCSP}{\ensuremath{\mathsf{Max}\text{-}\mathsf{CSP}}\xspace}
\newcommand{\MaxCSPs}{\ensuremath{\mathsf{Max}\text{-}\mathsf{CSPs}}\xspace}
\newcommand{\ExactCSP}{\ensuremath{\mathsf{Exact}\text{-}\mathsf{CSP}}\xspace}

\newcommand{\SPACE}{\ensuremath{\mathsf{SPACE}}\xspace}

\newcommand{\NRD}{\ensuremath{\mathsf{NRD}}\xspace}

\newcommand{\Core}{\ensuremath{\operatorname{core}}\xspace}
\newcommand{\Aut}{\ensuremath{\operatorname{Aut}}\xspace}
\newcommand{\Deg}{\ensuremath{\operatorname{deg}}\xspace}

\newcommand{\Val}{\ensuremath{\operatorname{val}}\xspace}
\newcommand{\UnsatVal}{\ensuremath{\operatorname{unsatval}}\xspace}
\newcommand{\OPTmax}{\ensuremath{\mathrm{OPT}_{\max}}\xspace}
\newcommand{\OPTmin}{\ensuremath{\mathrm{OPT}_{\min}}\xspace}

\newcommand{\Sat}{\ensuremath{\mathrm{Sat}}\xspace}
\newcommand{\Index}{\ensuremath{\mathsf{INDEX}}\xspace}
\newcommand{\CUT}{\ensuremath{\mathsf{CUT}}\xspace}
\newcommand{\MaxCut}{\ensuremath{\mathsf{Max}\text{-}\mathsf{Cut}}\xspace}
\newcommand{\MaxDiCut}{\ensuremath{\mathsf{Max}\text{-}\mathsf{DiCut}}\xspace}
\newcommand{\ThreeSAT}{\ensuremath{3\text{-}\mathsf{SAT}}\xspace}
\newcommand{\KSAT}{\ensuremath{k\text{-}\mathsf{SAT}}\xspace}

\newcommand{\NAE}{\ensuremath{\mathsf{NAE}}\xspace}

\title{Cores Characterize the One-Pass Streaming Complexity of \texorpdfstring{\CSPs}{CSPs}}
\author{%
\begin{tabular}{cc}
Joshua Brakensiek\thanks{Supported in part by NSF awards CCF-2211972 and DMS-2503280 and Venkatesan Guruswami's Simons Investigator award.}
& Aaron Putterman\thanks{Supported in part by a Jane Street Graduate Research Fellowship, Simons Investigator awards to Madhu Sudan and Salil Vadhan, and AFOSR award FA9550-25-1-0112.} \\
{\small University of California, Berkeley} & {\small Harvard University} \\
{\small \href{mailto:josh.brakensiek@berkeley.edu}{\nolinkurl{josh.brakensiek@berkeley.edu}}}
& {\small \href{mailto:aputterman@g.harvard.edu}{\nolinkurl{aputterman@g.harvard.edu}}} \\[1em]
Amatya Sharma
& Santhoshini Velusamy\thanks{Supported in part by NSERC Discovery grant RGPIN-2026-04689 and ECR supplement DGECR-2026-00421.} \\
{\small University of Michigan, Ann Arbor} & {\small University of Waterloo} \\
{\small \href{mailto:amatya@umich.edu}{\nolinkurl{amatya@umich.edu}}}
& {\small \href{mailto:santhoshini.velusamy@uwaterloo.ca}{\nolinkurl{santhoshini.velusamy@uwaterloo.ca}}}
\end{tabular}%
}
\date{}

\begin{document}
\maketitle

\begin{abstract}

We study the one-pass streaming complexity of CSPs over a fixed finite relation $R$ under three natural objectives: $\textsf{SAT}(R)$, deciding whether all constraints can be satisfied; $\textsf{Min}$-$\textsf{CSP}(R)$, approximately minimizing the number of unsatisfied constraints; and $\textsf{Exact-CSP}(R)$, exactly computing the maximum number of satisfied constraints.

Sharma and Velusamy (ESA 2026) characterized the streaming complexity of satisfiability for CSPs with literals using the non-redundancy parameter $\textsf{NRD}_n(R)$, which roughly measures the largest instance in which every constraint is independently necessary. For the most general setting without literals, they obtained the corresponding characterization for Boolean relations and showed obstacles to extending their techniques to larger domains. Kol, Paramonov, Saxena, and Yu (ITCS 2023) similarly characterized the streaming complexity of $\textsf{Exact-CSP}(R)$ for Boolean CSPs with literals in terms of the degree deg$(R)$ of the relation when written as a polynomial, while without literals, the corresponding result was known only for $\textsf{Max-Cut}$.

For all three of the objectives we study, we give tight characterizations for unweighted CSP instances over arbitrary finite domains. The streaming algorithms we provide are deterministic, while the corresponding bounds are tight up to polylogarithmic factors, even against randomized algorithms.
The central idea in all three results is a reduction of $R$ to its core, a canonical subrelation of $R$ which admits gadgets that allow for the hard-pinning (or fixing) of variables.
\end{abstract}

\pagenumbering{gobble}

\newpage

\tableofcontents

\newpage
\pagenumbering{arabic}

\section{Introduction}
\label{sec:intro}

Constraint satisfaction problems (\CSPs) are one of the basic objects in theoretical computer science. Fix a finite domain $D$, an arity $k$, and a relation $R\subseteq D^k$. An $R$-constraint applies $R$ to an ordered $k$-tuple of variables, with repetitions allowed, and an $R$-instance is a collection of such constraints. The basic satisfiability problem asks whether there is an assignment that satisfies every constraint. By varying the relation $R$, \CSPs can capture a huge variety of natural problems, from bipartiteness, to \ThreeSAT, to graph coloring, and to systems of linear equations.

More recently, a line of work has emerged which studies the complexity of solving such \CSPs in the \emph{streaming model}. In this setting, each of the $m$ constraints is revealed one at a time. The goal, as before, is to decide satisfiability or a related objective for the resulting instance using \emph{as little space} as possible. We assume throughout that the stream length is polynomially bounded in $n$. Naturally, storing every constraint explicitly as it arrives allows one to store the entire \CSP instance, meaning there is a trivial upper bound of $O(n^{k}\log n)$ bits of space for any such task.

The streaming study of \CSP optimization was initiated by Kapralov, Khanna, and Sudan~\cite{KKS15}, who proved strong lower bounds for approximating \MaxCut. Kapralov and Krachun~\cite{KK19} subsequently proved that there is no non-trivial approximation for \MaxCut in sublinear-space in the streaming model. 
Guruswami, Velingker, and Velusamy~\cite{GVV17} and Chou, Golovnev, and Velusamy~\cite{CGV20} studied Boolean \MaxCSPs of arity two, while Chou, Golovnev, Sudan, and Velusamy~\cite{CGSV24} extended this line of work to arbitrary finite \MaxCSPs.

Very recently, this line of work has led to substantially broader characterizations of the streaming approximability of \MaxCSPs. A sequence of works developed increasingly strong sublinear-space algorithms and lower bounds for \MaxDiCut and \MaxCut~\cite{SSSV25,FMW25,Vel25,ABFS26}, while Singer, Tulsiani, and Velusamy \cite{STV25} and Fei, Minzer, and Wang \cite{FMW26} uncovered a close connection between streaming approximability and the standard BasicLP relaxation for \CSPs. Subsequent works essentially completed this picture in the one-pass setting, giving sublinear-space algorithms up to the approximation threshold predicted by BasicLP and linear-space lower bounds beyond it~\cite{ABF26,STV26,FMW26-tight}. Thus, the streaming approximability of \MaxCSPs is now understood to a remarkable extent, in sharp contrast to the objectives that we study in this work.

Kol, Paramonov, Saxena, and Yu~\cite{KPSY23} studied exact optimization for Boolean \CSPs with literals. Their formulation asks, for a given threshold $v$, whether some assignment satisfies at least $v$ constraints, which is equivalent to computing the maximum number of satisfied constraints exactly. They showed that if the indicator of the predicate has polynomial degree $d$, then the streaming space complexity is $\widetilde{\Theta}(n^d)$. Other work has focused on satisfiability. Vu~\cite{Vu-2208.09160} proved an $\Omega(n^k)$ lower bound for \KSAT, while Chou, Golovnev, Sudan, and Velusamy~\cite{CGSV24} proved an $\Omega(n)$ lower bound for arbitrary \CSPs.

The streaming complexity of satisfiability was characterized more recently by Sharma and Velusamy~\cite{sharma2026characterizingstreamingdecidabilitycsps}, who showed that for a relation $R$, its space complexity is characterized, up to a logarithmic factor, by the non-redundancy of $R$. Non-redundancy in turn is a fundamental parameter in the study of \CSPs, known to govern the query complexity of learning \CSPs \cite{bessiere2020chain, carbonnel2022redundancy}, the kernelizability of \CSPs \cite{jansen2019optimal,chen2020best,lagerkvist2020sparsification}, and approximate sparsification of \CSPs \cite{kogan2015sketching,filtser2017sparsification,BZ20,KPS24,khanna2025efficient,brakensiek2025redundancy}. Given this rapidly-growing list of connections, there has been a recent surge of interest in obtaining sharp asymptotic estimates for the non-redundancy of many \CSP predicates~\cite{bessiere2020chain,chen2020best,lagerkvist2020sparsification,carbonnel2022redundancy,khanna2025efficient,brakensiek2025redundancy,brakensiek2025richness,LGE26,BrakensiekGP26,SV26b,brakensiek2026shrinking}, often combining ideas from logic, algebra, and extremal combinatorics. Recent work has obtained further bounds and structural results for non-redundancy and sparsification~\cite{brakensiek2026contraction,brakensiek2026separating,brakensiek2026abelian}.

These characterizations rely crucially on the use of literals. In the general-domain literal model, the domain is $\mathbb Z_q$ and each occurrence of a variable may be replaced by $x+\lambda$ for an arbitrary shift $\lambda\in\mathbb Z_q$. Over the Boolean domain, these are precisely the two literals $x$ and $\neg x$\footnote{By the general finite-domain setting, we mean that the relation and its finite domain are arbitrary, and constraints apply the relation directly to variables, without literals or shifts.}.
Literals are more than a notational convenience in these arguments. Indeed, negations or additive shifts allow for much more powerful gadgets which lead to simpler lower bound proofs.
Sharma and Velusamy \cite{sharma2026characterizingstreamingdecidabilitycsps} showed that this is a fundamental obstruction in applying their techniques to general finite-domain \CSPs. Over domains larger than two, the raw non-redundancy of $R$ need not predict the streaming complexity of \CSP satisfiability.

The difficulty lies specifically in proving lower bounds. The known upper bounds continue to apply, since ordinary \CSP instances are special cases of instances in the literal model. In the lower bound reductions however, literals were critical in gadget constructions that fixed the values of certain variables. 
Our starting question is therefore the following.
\begin{center}
    \emph{Can we derive tight characterizations of the streaming complexity of \CSPs?%
    }
\end{center}

As our primary contribution, we give an affirmative answer to the above question, in part making use of a crucial object called the \emph{core} of the \CSP relation $R$, denoted by $\Core(R)$, which intuitively removes domain values that can be collapsed without changing any of these objectives. In the context of \CSPs, the notion of a core originated in the 1970s in the context of designing optimal queries to a database~\cite{ChandraM77,KolaitisV08} but the concept was also independently developed in the graph theory literature (see \cite{HellN92,HellN08}). Perhaps the most notable use of \CSP cores is their key role in the proof of the \CSP dichotomy theorem~\cite{feder1998computational,BartoKW17,Bulatov17,zhuk2020proof}. Despite these many connections, our work appears to be the first to utilize the concept of a \CSP core in the context of streaming.

Beyond simply studying $\SAT(R)$, which asks whether all constraints can be satisfied, we also study $\MinCSP(R)$, which approximately minimizes the number of unsatisfied constraints, and $\ExactCSP(R)$, which computes the maximum number of satisfied constraints exactly. Thus, $\ExactCSP(R)$ is the exact-value version of $\MaxCSP(R)$, and we use the latter terminology primarily for the approximation literature. As a byproduct of our study, we show that (a) non-redundancy is even more entrenched as a crucial parameter which governs the streaming complexity of \CSP problems and (b) the core is a fundamental object which allows for the extension of lower bound techniques to general finite-domain \CSPs.

\subsection{Our Results}

We study the one-pass streaming complexity of solving three different CSP objectives defined by a fixed finite relation $R$, namely $\SAT(R)$, $\MinCSP(R)$, and $\ExactCSP(R)$.
The main parameter that governs the space complexity in our results is non-redundancy. Informally, an instance is non-redundant if every constraint is non-redundant, i.e., removing it would strictly grow the set of satisfying assignments of the instance. This is equivalent to having a witness assignment for each constraint that violates it but satisfies the rest of the constraints. We write $\NRD_n(R)$ for the largest number of constraints in such an instance on $n$ variables. Our first theorem takes advantage of both $\NRD$ and $\Core(R)$ and characterizes the randomized one-pass streaming complexity of $\SAT(R)$ up to a $\log$ factor. 

\begin{theorem}[Informal \Cref{thm:sat-main} satisfiability characterization]
For every fixed non-trivial relation $R$, there is a deterministic one-pass algorithm that solves $\SAT(R)$ using $O_R(\NRD_n(\Core(R))\log n)$ space, and any one-pass algorithm (even randomized) requires at least $\Omega_R(\NRD_n(\Core(R)))$ space.

\end{theorem}

Our minimization result comes from a stronger sparsification theorem, resembling those of \cite{KPS24, khanna2025efficient, brakensiek2025redundancy}. The goal here is, given a \CSP instance $C$, to build a small, re-weighted sub-instance (denoted $\hat{C}$) such that on \emph{every assignment} $(x_1,\ldots,x_n)\in D^n$, the number of unsatisfied constraints in $\hat{C}$ is essentially the same as the number of unsatisfied constraints in $C$. Clearly, once such a sparsifier is constructed, one can (approximately) solve $\MinCSP(R)$. Importantly, we show that building and maintaining such sparsifiers can still be done in space roughly equal to the non-redundancy:

\begin{theorem}[Informal \Cref{thm:streaming-sparsification} and \Cref{cor:mincsp-main} sparsification and minimization characterization]
For every fixed non-trivial relation $R$ and every $\varepsilon>0$, there is a deterministic one-pass algorithm that can store a $(1\pm\varepsilon)$ sparsifier using $O_{R,\varepsilon}(\NRD_n(R)\poly(\log n))$ space. Consequently, there is a deterministic one-pass $(1+\varepsilon)$-approximation algorithm for $\MinCSP(R)$ using $O_{R,\varepsilon}(\NRD_n(\Core(R))\poly(\log n))$ space, and every finite-factor randomized approximation requires $\Omega_R(\NRD_n(\Core(R)))$ space.
\end{theorem}

Finally, for $\ExactCSP$, the space complexity is determined by the polynomial degree of the core\footnote{More formally, if $C$ is the domain of the core, we write our predicate $\Core(R)$ as a real-valued indicator function $f : C^k \to \mathbb R$, compute its Fourier decomposition, and let $d$ be degree of the highest-degree monomial with a non-zero coefficient.}.

\begin{theorem}[Informal \Cref{thm:exact-main} $\ExactCSP$ characterization]
Let $d$ be the degree of the indicator function of $\Core(R)$. Solving $\ExactCSP(R)$ in one pass on simple instances requires $\Omega_R(n^d)$ randomized space and has a deterministic $O_R(n^d\log n)$-space algorithm.
\end{theorem}

Thus all three problems first reduce to the core. After that reduction, non-redundancy characterizes satisfiability and minimization up to polylogarithmic factors, while degree characterizes $\ExactCSP$ up to a logarithmic factor.

\paragraph{The Core}

The reason that $\Core(R)$ is central to each of these results is due to the core being the \emph{canonical fixed-point} of $R$ under homomorphisms. More precisely, given a finite domain $D$ and a relation $R \subseteq D^k$, a homomorphism of $R$ to itself is a map $h : D \to D$ such that for every tuple $(t_1, \hdots, t_k) \in R$, we have that $(h(t_1), \hdots, h(t_k)) \in R$. Given an $R$-instance of a \CSP, a homomorphism $h$ always maps a satisfying assignment to a satisfying assignment.
Furthermore, in the optimization setting, applying a homomorphism to an arbitrary assignment cannot decrease the number of satisfied constraints.
Therefore, in any of our streaming problems, we may apply any homomorphism we desire at no cost.

If the map $h$ is not a bijection, then $h$ maps $R$ to a smaller-domain relation $h(R)$. Since the domain $D$ is finite, such a shrinking can only be done a finite number of times until a fixed point is reached--this is precisely the core of the relation. As such, for any of our streaming problems, an algorithm for $\Core(R)$ is a natural upper bound on the problem. A corresponding lower bound follows from the fact that \CSP cores have canonical gadgets that can be used to fix variables. Further explanation and concrete examples are given in the technical overview (\cref{sec:overview}).

\begin{samepage}
\section*{AI Disclosure.}
Generative AI tools including GPT $5.6$ Sol Medium/High were used for language editing, improving exposition, checking the clarity of parts of the manuscript and to assist with routine mathematical tasks that do not require a significant level of innovation or with verifying proofs. They were not used to generate the mathematical results or proofs, and the authors take full responsibility for the contents of the paper.
\par
\end{samepage}

\section{Technical overview}
\label{sec:overview}

In this section, we explain the main structural ideas behind our results. Our central idea is to replace every relation by a reduced form called its \emph{core}, which we define formally in \Cref{subsec:core-prelims}. We first use a relation of Sharma and Velusamy~\cite{sharma2026characterizingstreamingdecidabilitycsps} to explain why we must reduce to the core. We then explain the satisfiability lower bound using $\CUT = \{(0,1),(1,0)\}$ and return to $\CUT$ to explain the exact optimization lower bound.  As we will show, the satisfiability and exact-optimization lower bounds crucially rely on whether we can \emph{hard-pin} variables, meaning fix them to chosen values in the domain.
$R$-constraints usually cannot do this directly, so we instead build gadgets that simulate such pinnings up to unavoidable symmetries of the relation. The main fact that makes these gadgets possible is that the relation is a core. For satisfiability, we need a gadget whose satisfying assignments give the variables the desired values. For exact optimization, an optimal assignment may violate some gadget constraints, so we need a stronger construction that makes every deviation from the desired values unprofitable.

\subsection{Reduction to the core}

Sharma and Velusamy~\cite{sharma2026characterizingstreamingdecidabilitycsps} proved that the one-pass space needed to decide satisfiability for \CSPs with literals is characterized, up to a logarithmic factor, by the \emph{non-redundancy} parameter $\NRD_n(R)$. This parameter is the maximum number of constraints in an $R$-instance in which every constraint is independently necessary. Equivalently, every constraint has a \emph{witness assignment} that violates it and satisfies all the other constraints. We give the formal definition in \Cref{subsec:nrd}.

They also showed that without literals, $\NRD_n(R)$ need not be the right parameter, using the following counterexample
\[
T=\{(0,1),(1,0),(0,3),(3,0),(1,2),(2,1)\}\subseteq\{0,1,2,3\}^2.
\]
Using known characterizations of non-redundancy, one can prove that $\NRD_n(T)=\Theta(n^2)$ \cite{BZ20,bessiere2020chain,khanna2025efficient,sharma2026characterizingstreamingdecidabilitycsps}. However, satisfiability for $T$ can be checked in $O(n\log n)$ space. To see this, notice that $T(a,b)$ holds when $a$ and $b$ have opposite parity and at least one of them belongs to $\{0,1\}$. Given a $T$-instance, form its underlying graph whose vertices are the variables and whose edges are the pairs $\{x,y\}$ on which a constraint $T(x,y)$ appears. %
The $T$-instance is satisfiable exactly when this graph is bipartite. Indeed, every satisfying assignment gives a bipartition according to the parity of the assigned values. Conversely, any bipartition can be made satisfying by assigning $0$ to one side and $1$ to the other. Bipartiteness can be decided in one pass using  $O(n\log n)$ bits \cite{sharma2026characterizingstreamingdecidabilitycsps}.

This suggests that the values $2$ and $3$ are unnecessary for solving the instance. So, let's consider the map
\[
\sigma(0)=\sigma(2)=0
\qquad\text{and}\qquad
\sigma(1)=\sigma(3)=1.
\]
If $(a,b)\in T$, then $a$ and $b$ have opposite parity. Therefore $(\sigma(a),\sigma(b))$ is either $(0,1)$ or $(1,0)$ and still belongs to $T$. Applying $\sigma$ to an assignment therefore preserves every constraint that was already satisfied. Thus every assignment can be replaced by a Boolean assignment satisfying at least the same constraints.

The restriction of $T$ to $\{0,1\}$ is precisely $\CUT=\{(0,1),(1,0)\}$. In this sense, $T$ is essentially the same relation as $\CUT$. Every assignment over the original domain can be mapped to a Boolean assignment without losing a satisfied constraint, while every Boolean assignment is already available over the original domain. It follows that restricting the domain to $\{0,1\}$ preserves satisfiability, the minimum number of unsatisfied constraints, and the maximum number of satisfied constraints. This is also the reason the bipartiteness argument above works, but the map $\sigma$ states it in a form that extends to every relation.

Formally, a map that preserves a relation in this way is called an \emph{endomorphism}, defined formally in \Cref{subsec:core-prelims}. The \emph{core} of a relation is obtained by choosing an endomorphism with the smallest possible image and restricting the relation to that image. In this example, $\sigma$ has image $\{0,1\}$, and no endomorphism can have a one-element image because $T$ contains no pair $(a,a)$. We therefore have $\Core(T)=\CUT$.

Notice that $\NRD_n(\CUT)=\Theta(n)$. If a $\CUT$-instance is satisfiable and nonredundant, its graph must be a forest because every edge on a cycle is implied by the remaining edges. If it is unsatisfiable and nonredundant, removing any edge must make it bipartite, so its nonisolated part is an odd cycle. In either case it has $O(n)$ edges, while a forest gives the matching lower bound. Reducing to the core therefore replaces the incorrect quadratic parameter $\NRD_n(T)$ by the correct linear parameter $\NRD_n(\CUT)$.

The examples above suggest the following general statement, which we prove formally in \Cref{sec:cores}.

\begin{theorem}[Informal \Cref{lem:core-reduction,lem:core-space}]
For every relation $R$, replacing each $R$-constraint by the $\Core(R)$-constraint on the same scope preserves satisfiability, the minimum number of unsatisfied constraints, and the maximum number of satisfied constraints. Consequently, this replacement preserves the one-pass streaming space of all three problems we consider.
\end{theorem}

Moreover, this reduction can be performed as the stream arrives. Whenever an $R$-constraint appears, we replace it by the $\Core(R)$-constraint on the same scope and give it to the algorithm for the core relation.

Now, the remaining task is to prove tight space bounds for the core relation. The satisfiability upper bound now follows immediately. We reduce the stream to the core and apply the algorithm of Sharma and Velusamy~\cite{sharma2026characterizingstreamingdecidabilitycsps}, which uses $O_R(\NRD_n(\Core(R))\log n)$ space. As we will see later in \Cref{subsec:over-exact}, the upper bound for exact-optimization follows as well using a simple argument. It remains to prove the lower bound matching this up to the logarithmic factor.

For both the lower bounds, we use the standard connection between streaming algorithms and one-way communication together with the hard problem $\Index$, where Alice receives a string $b\in\{0,1\}^m$ and Bob receives an index $j$ and must recover $b_j$. 
At a high level, we reduce this communication problem to a streaming problem as follows. 
Alice runs the streaming algorithm on constraints determined by her input and sends its memory state to Bob. Bob then appends constraints determined by his query and uses the output to recover the answer. The difficulty is that every constraint Bob appends must be a constraint of the core relation. He cannot directly fix variables to the values determined by his query. The next subsection explains how the core property lets him simulate this operation for satisfiability. The corresponding construction for exact optimization will require a stronger version of the same idea.

\subsection{Satisfiability and hard pinning}

We now explain why every one-pass algorithm for satisfiability over a core relation $R$ requires $\Omega_R(\NRD_n(R))$ space, as stated formally in \Cref{thm:sat-core-lb}. The proof first gives an ideal reduction in which Bob can fix variables to chosen values. We then use $\CUT$ to see why unique fixing is sometimes impossible and what weaker property is enough. Finally, we state the general hard-pinning lemma for core relations.

Let $I=\{C_1,\ldots,C_m\}$ be a maximum \emph{nonredundant instance} over a core relation $R$, where $m=\NRD_n(R)$. Recall from \Cref{subsec:nrd} that every constraint $C_j$ has a \emph{witness assignment} $\alpha_j$ that violates $C_j$ and satisfies every other constraint in $I$.

In \Index, Alice receives a string $b\in\{0,1\}^m$ and inserts $C_i$ whenever $b_i=1$. Bob receives an index $j$ and must recover $b_j$. If Bob could append constraints that \emph{hard-pin} every variable to the value given by $\alpha_j$, meaning that the appended constraints force this assignment, then the reduction would be immediate. If $b_j=0$, the witness $\alpha_j$ satisfies Alice's entire instance. If $b_j=1$, Alice also inserted $C_j$, which $\alpha_j$ violates. Thus the final instance would be satisfiable exactly when $b_j=0$.

The obstacle is that $R$-constraints may not be able to force a unique assignment. The following example shows both the obstacle and the weaker form of pinning that we need.

\begin{example}[Hard pinning for $\CUT$]
Consider $R=\CUT$. Complementing every value preserves every constraint, so any instance satisfied by $\alpha_j$ is also satisfied by its complement. Bob does not need uniqueness, however. It is enough to append an instance whose satisfying assignments on the original variables are $\alpha_j$ and its complement.

Bob introduces two dummy variables $z_0,z_1$ and adds the constraint $\CUT(z_0,z_1)$. For every original variable $x$, he adds $\CUT(x,z_1)$ if $\alpha_j(x)=0$ and $\CUT(x,z_0)$ if $\alpha_j(x)=1$. Any assignment satisfying these constraints must give $z_0$ and $z_1$ different values. Once those two values are chosen, every original variable is forced. The assignment on the original variables is therefore either $\alpha_j$ or its complement.

This completes the reduction for $\CUT$. If $b_j=0$, the assignment $\alpha_j$ satisfies Alice's constraints and Bob's constraints. If $b_j=1$, Alice inserted $C_j$, which is violated by $\alpha_j$. Its complement also violates $C_j$ because complementing every value preserves whether a $\CUT$-constraint is satisfied. Hence the final instance is satisfiable exactly when $b_j=0$.

The map that swaps $0$ and $1$ is an endomorphism of $\CUT$ that only permutes its values. Such a bijective endomorphism is called an \emph{automorphism}. Applying an automorphism to every value preserves every constraint, which explains why the unavoidable ambiguity between $\alpha_j$ and its complement is harmless.
\end{example}

The $\CUT$ example shows that Bob does not need to force one assignment uniquely. It is enough to force all original variables according to the same symmetry of the relation. For a core relation, every endomorphism is an automorphism, and this is precisely the property that makes the general construction work.

\begin{lemma}[Informal \Cref{lem:hard-pinning}]
Let $R$ be a core relation and let $\alpha$ be an assignment to a set of variables. There is an $R$-instance whose satisfying assignments on these variables are exactly the assignments $\pi\circ\alpha$, where the same automorphism $\pi$ of $R$ is applied to every variable.
\end{lemma}

Here is the idea behind the lemma. Introduce one dummy variable $y_a$ for every value $a$ in the domain. Add constraints among these variables so that, under any satisfying assignment, the map sending $a$ to the value of $y_a$ preserves $R$ and is therefore an endomorphism. Since $R$ is a core relation, this map must be an automorphism. We then attach every original variable $x$ to these dummy variables so that $x$ must receive the image of $\alpha(x)$ under the same automorphism. Thus all original variables follow one common symmetry of $R$.

Every assignment $\pi\circ\alpha_j$ satisfies and violates exactly the same $R$-constraints as $\alpha_j$. Bob can therefore append the instance from the informal lemma in the ideal reduction above. If $b_j=0$, an automorphic image of $\alpha_j$ satisfies the final instance. If $b_j=1$, every assignment satisfying Bob's constraints violates $C_j$. A streaming algorithm for satisfiability would therefore give a one-way protocol for $\Index_m$, proving the $\Omega_R(\NRD_n(R))$ lower bound for core relations. Reducing an arbitrary relation to its core gives the characterization in \Cref{sec:sat}.

For the sake of brevity, we next discuss the techniques that underlie the Exact-CSP results. For a discussion of the sparsification and minimization results, we encourage the reader to look at \cref{sec:mincsp}.

\subsection{\texorpdfstring{\ExactCSP}{Exact-CSP}}\label{subsec:over-exact}

We now explain why exact optimization requires a stronger pinning construction than satisfiability. Let $d$ be the \emph{degree} of the indicator function of the core relation. Informally, this is the largest number of coordinates that occur together in a monomial of its polynomial representation. Our result gives an $O_R(n^d\log n)$ upper bound and a matching $\Omega_R(n^d)$ lower bound up to the logarithmic factor. The formal statement is \Cref{thm:exact-main}, and the formal definition of degree is in \Cref{subsec:degree}.

As before, we first reduce the relation to its core. The upper bound stores the coefficients of the low-degree polynomial representing the value of the instance. For the lower bound, the ideal reduction again asks Bob to hard-pin complete assignments. It is tempting to use \Cref{lem:hard-pinning} for this purpose. But that lemma is not strong enough here. It describes assignments that satisfy every pinning constraint, whereas an optimal assignment may violate some pinning constraints in order to satisfy more of Alice's constraints. We therefore first describe the ideal reduction and then explain how a stronger pinning construction ensures that such a violation cannot improve the objective.

\begin{example}[Degree of $\CUT$]
For $\CUT$ on the domain $\{0,1\}$, the indicator is $\mathbf 1_{\CUT}(x,y)=x+y-2xy$. Its degree is therefore $2$.
\end{example}

\begin{example}[Degree of $\NAE$]
Let $R_{\NAE}\subseteq\{0,1\}^3$ contain the tuples whose three entries are not all equal. Its indicator is $\mathbf 1_{R_{\NAE}}(x,y,z)=x+y+z-xy-xz-yz$. Thus $R_{\NAE}$ has degree $2$ even though its arity is $3$.
\end{example}

\paragraph{Upper bound.}

The upper bound follows directly from the polynomial representation. Kol, Paramonov, Saxena, and Yu~\cite{KPSY23} give the Boolean version. They write the value of an instance as a polynomial in the variables and store its coefficients. Each arriving constraint adds one fixed low-degree polynomial to this coefficient table. We use the same idea over an arbitrary finite domain, using the value indicators from \Cref{subsec:degree} in place of Boolean variables.

\begin{example}[The upper bound for $\CUT$]
Consider a stream of edges defining a $\CUT$-instance. By the expansion above, an arriving edge $\{x_u,x_v\}$ contributes $x_u+x_v-2x_ux_v$ to the value of every assignment. The algorithm adds these three coefficients to its table. After the stream, the sum of the stored terms is exactly the number of edges cut by every assignment. There are only $O(n^2)$ possible terms, and their coefficients use $O(\log n)$ bits on a simple instance.  
The algorithm may then try every assignment, evaluate the stored polynomial, and return its maximum. Thus $\ExactCSP(\CUT)$ uses $O(n^2\log n)$ space.
\end{example}

The same argument works over every finite domain. If the indicator of the core relation has degree $d$, each arriving constraint contributes a fixed polynomial whose terms involve at most $d$ variables. Across all scopes, there are only $O_R(n^d)$ possible terms. The algorithm stores the coefficient of each term, and after the stream it tries all assignments and returns the largest value. This gives the $O_R(n^d\log n)$ upper bound and extends the Boolean argument of Kol, Paramonov, Saxena, and Yu to arbitrary finite domains.

\paragraph{The lower bound.}
We again reduce from one-way $\Index$. We call the values of Alice's instance at assignments chosen by Bob \emph{point evaluations}. The following $\CUT$ example shows how degree lets Bob combine these point evaluations to recover his bit.

\begin{example}[Four point evaluations for $\CUT$]
Alice receives bits $b_{ij}$ indexed by $(i,j)\in[N]^2$. She creates vertices $x_1,\ldots,x_N$ and $y_1,\ldots,y_N$, and inserts the edge $\{x_i,y_j\}$ whenever $b_{ij}=1$. Let $P$ denote the resulting graph.

Bob receives a pair $(p,q)$. He uses four assignments. The assignment $s_{00}$ sets every variable to $0$. The assignment $s_{10}$ sets only $x_p$ to $1$, the assignment $s_{01}$ sets only $y_q$ to $1$, and the assignment $s_{11}$ sets only $x_p$ and $y_q$ to $1$. In the ideal reduction, Bob makes four copies of Alice's streaming state and appends a different input to each copy, one for each of these assignments. The four point evaluations satisfy
\[
\Val_P(s_{11})-\Val_P(s_{10})-\Val_P(s_{01})+\Val_P(s_{00})=-2b_{pq}.
\]

To see the cancellation, consider one edge of $P$. An edge not incident to $x_p$ or $y_q$ contributes zero to the signed sum. An edge incident to exactly one of them contributes equally to two terms with opposite signs. The only edge that can remain is $\{x_p,y_q\}$. If it is present, it is cut under $s_{10}$ and $s_{01}$ but not under $s_{00}$ or $s_{11}$, giving a total contribution of $-2$.
\end{example}

Thus Bob can recover $b_{pq}$ if the run from each copy reveals the value of $P$ at its chosen assignment. If Bob could directly hard-pin the chosen assignment, he could append the required constraints in each run. For $\CUT$, it is enough to force either that assignment or its complement, since complementation preserves the value of every $\CUT$-instance. The optimum of the resulting instance, after subtracting the known number of pinning constraints, would reveal the required point evaluation.

Roughly, the same ideal reduction extends to every core relation of degree $d$. Alice's bits are indexed by $[N]^d$, and Bob considers $2^d$ assignments obtained by changing the selected variable in each of $d$ groups. A signed sum of their values cancels every constraint except the one indexed by Bob's query. If Bob could hard-pin each assignment in one run of the streaming algorithm, these $2^d$ optimum values would reveal his bit.

It would be tempting to use the same satisfiability hard-pinning lemma, \Cref{lem:hard-pinning}. However, it only controls assignments that satisfy the whole pinning instance. An optimum for the exact-optimization problem might instead violate one pinning constraint and gain many constraints from Alice's instance. We therefore replace the constant-size pinning gadget by large blocks whose constraints outweigh every possible gain from $P$.

\begin{example}[Strong pinning for $\CUT$]
The graph $P$ has $2N$ vertices, at most $N^2$ edges, and maximum degree at most $N$. Fix an assignment $s$ to its vertices. We want to add a gadget that forces every optimum to agree with $s$ on the original vertices, up to complementing every value. To do this, we introduce extra vertices. Let $L=1000N$, and create two blocks $Z_0,Z_1$, each containing $L$ vertices. Add every edge between the two blocks. For every original vertex $x$, add $10N$ edges from $x$ to distinct vertices in $Z_1$ when $s(x)=0$, and to distinct vertices in $Z_0$ when $s(x)=1$.

We first show that every optimum makes the two blocks monochromatic, meaning that all vertices in each block receive the same value, and gives the blocks opposite values. Choose a most frequent value in each block. If these two values are the same, at least $L^2/4$ edges between the blocks are uncut. Making the blocks monochromatic with opposite values cuts all block edges and loses at most the $20N^2$ edges attached to the original vertices. Since $L^2/4>20N^2$, an optimum must have opposite majority values. Up to complementing the entire cut, assume that the majority value of $Z_0$ is $0$ and that of $Z_1$ is $1$.

Let $a$ be the number of vertices assigned $1$ in $Z_0$, and let $c$ be the number assigned $0$ in $Z_1$. The number of uncut edges between the blocks is $a(L-c)+(L-a)c=L(a+c)-2ac\ge L(a+c)/2$. Making both blocks monochromatic according to their majority values gains these edges and loses at most $2N(a+c)$ attached edges. Since $L/2>2N$, an optimum must have $a=c=0$. Thus $Z_0$ is entirely $0$ and $Z_1$ is entirely $1$, up to complementing all values.

Once the blocks are fixed, every original vertex $x$ receives all its $10N$ attached edges exactly when $x=s(x)$. Giving $x$ the other value loses $10N$ attached edges and can gain at most its $N$ incident edges in $P$. Hence every optimum agrees with $s$ on all original vertices, up to complementing the whole assignment.

Let $H_s$ denote the block and attached edges. The intended assignment cuts every edge of $H_s$, and complementation preserves the value of every $\CUT$-instance. Therefore
\[
\OPTmax(P\cup H_s)=|H_s|+\Val_P(s).
\]
The construction depends only on the public value $N$ and on $s$, so Bob can build it without knowing which edges Alice inserted.
\end{example}

Bob now makes four copies of Alice's streaming state and appends $H_{s_{00}}$, $H_{s_{10}}$, $H_{s_{01}}$, and $H_{s_{11}}$ to the respective copies. From each optimum, he subtracts the known size of the appended gadget and obtains one point evaluation of $P$. Their alternating sum reveals $b_{pq}$, giving an $\Omega(N^2)$ lower bound for $\ExactCSP(\CUT)$.\footnote{Kol, Paramonov, Saxena, and Yu~\cite{KPSY23} also prove an ad hoc $\Omega(n^2)$ lower bound for Max-Cut via a different reduction.}

For a general core relation of degree $d$, the construction in \Cref{lem:exact-pinning} plays the same role. It replaces the two $\CUT$ blocks by one block for every value in the core domain. The constraints among the blocks force their most frequent values to define one common endomorphism, which must be an automorphism because the relation is a core. The attached constraints then force the original variables to follow the desired assignment under this automorphism. As in the $\CUT$ example, the blocks are large enough that every deviation loses more attached constraints than it can gain in Alice's instance.

The formal construction gives exactly the property needed in the communication reduction.

\begin{lemma}[Informal \Cref{lem:exact-pinning}]
Let $R$ be a core relation, and let $P$ be a simple $R$-instance on $O_R(N)$ variables with at most $N^d$ constraints and maximum variable degree at most $N^{d-1}$. For every assignment $s$ to the variables of $P$, there is an $R$-instance $H_s$, using $O_R(N)$ additional variables and constructible without knowing the constraints of $P$, such that $\OPTmax(P\cup H_s)=|H_s|+\Val_P(s)$.
\end{lemma}

Alice's instance and each pinning gadget use only $O_R(N)$ variables. Bob can therefore recover the $2^d$ point evaluations needed by the signed sum, and the lower bound for $\Index_{N^d}$ gives the $\Omega_R(n^d)$ space lower bound.

The two lower bounds therefore use the same structural principle. We first reduce exactly to the core and then recover the variable pinnings. For satisfiability, it is enough to restrict the satisfying assignments up to one common automorphism. For exact optimization, the restriction must additionally be strong enough that violating it is never profitable. The sparsification result uses the same core reduction to obtain the minimization characterization.%

\subsection{Relation to prior work}

We first explain how our results extend the earlier work of Sharma and Velusamy~\cite{sharma2026characterizingstreamingdecidabilitycsps}. Here, a constraint language is a fixed finite set of relations over the same domain. They proved the satisfiability characterization in two settings. First, they considered every finite domain but allowed literals: when the domain is $\mathbb{Z}_q$, each occurrence of a variable $x$ in a constraint may be replaced by $x+\lambda$ modulo $q$, for an arbitrary $\lambda\in\mathbb{Z}_q$. In the Boolean setting, this is exactly the choice between $x$ and $\neg x$. These shifts give the lower bound construction more freedom. For a fixed relation, a lower bound that uses nonzero shifts need not apply when all shifts are zero.

Their second result concerns positive Boolean \CSPs, meaning Boolean \CSPs without literals. They assume that the constraint language is not ``constant-satisfiable''. This means that there is no value $c\in\{0,1\}$ for which assigning every variable the value $c$ satisfies every relation in the language. This is the general interesting case for positive Boolean \CSPs, since otherwise every instance has the same constant satisfying assignment.

\begin{lemma}[Informal \Cref{lem:shift-core-nrd,lem:boolean-languages-are-cores}]
Reducing a language with literals to its core does not change its non-redundancy. Moreover, a Boolean constraint language is a core if and only if it is not constant-satisfiable.
\end{lemma}

Our core viewpoint uses these two facts to explain both earlier results. For the literal model, treat every shifted use of a relation as a separate relation. The resulting language is shift-closed, so the first part of the lemma shows that passing to its core does not change its non-redundancy. Our general core characterization therefore gives the earlier characterization for \CSPs with literals. For the positive Boolean result, the only non-bijective maps on the Boolean domain are the two constant maps, which send both domain values to $0$ or both to $1$. The second part of the lemma therefore shows that the assumption of not being constant-satisfiable is exactly the assumption that the language is a core.

Our work then goes beyond both results by giving a characterization in the most general finite-domain setting.
The precise statements and their proofs are given in \Cref{sec:literals-and-constants}.

We also compare our exact-optimization result with the work of Kol, Paramonov, Saxena, and Yu~\cite{KPSY23}. Their general characterization is restricted to Boolean \CSPs with literals and constants. Our result applies to arbitrary relations, requiring neither literals nor constants, giving a characterization in the more general \CSP setting. However, their lower bound does extend to a constant number of passes. We view an extension of our result to the multi-pass setting as an exciting future direction.

\section{Preliminaries}
\label{sec:prelims}

In this section, we state the definitions used in the rest of the paper. %

\subsection{\texorpdfstring{\CSPs}{CSPs}}
\label{subsec:csp-prelims}

Throughout the paper, $D$ is a fixed finite nonempty domain and $R\subseteq D^k$ is a fixed relation of arity $k\ge 1$. We work with ordinary $R$-constraints, without literals or shifts. We refer to this as the general finite-domain setting, since the finite domain and relation are arbitrary. We call $R$ \emph{non-trivial} if $\varnothing\subsetneq R\subsetneq D^k$. The asymptotic symbols $O_R$, $\Omega_R$, and $\Theta_R$ hide constants that may depend on $R$.%

Let $V=\{x_1,\ldots,x_n\}$ be a set of variables. An \emph{$R$-constraint} is an expression $R(x_{i_1},\ldots,x_{i_k})$, and its \emph{scope} is the ordered variable tuple $(x_{i_1},\ldots,x_{i_k})$. The order matters unless $R$ is symmetric. Variables may repeat within a scope. An assignment $\alpha:V\to D$ satisfies this constraint if $(\alpha(x_{i_1}),\ldots,\alpha(x_{i_k}))\in R$.

An unweighted $R$-instance $I$ is a finite multiset of $R$-constraints. Thus the same constraint may occur more than once, and $|I|$ counts constraints with multiplicity. A weighted instance gives every occurrence of a constraint $C$ a nonnegative rational weight $w(C)$ represented in binary. Unweighted instances are the special case in which every weight is $1$. The support of a weighted instance is the set of constraint occurrences with nonzero weight. A weighted subinstance $H\subseteq I$ uses only occurrences from $I$, but it may give them new nonnegative weights.

An instance is \emph{simple} if no ordered scope occurs more than once. The degree of a variable is the number of constraint occurrences whose scopes contain it. A constraint contributes once to this degree even if the variable appears more than once in its scope. The maximum variable degree is denoted by $\Delta(I)$. A relation is \emph{symmetric} if membership in $R$ is unchanged by permuting the $k$ coordinates. For a symmetric relation, we also use the stronger notion of \emph{unordered simplicity}. This means that no two scopes become the same multiset of variables after their coordinates are reordered. %

\begin{example}
Let $D=\{0,1\}$ and $\CUT=\{(0,1),(1,0)\}$. A $\CUT$-constraint $\CUT(x,y)$ requires $x$ and $y$ to receive different values. Thus a $\CUT$-instance is a graph, with one variable at each vertex and one constraint at each edge.
\end{example}

\subsection{Objectives and streaming space}
\label{subsec:objectives}

We next define the three objectives studied in the paper. We state them for a weighted $R$-instance, which is the more general setting. The inputs to our three problems are unweighted unless we say otherwise. Let $I$ be a weighted $R$-instance and let $\alpha$ be an assignment to its variables. We write
\[
\Val_I(\alpha)=\sum_{\substack{C\in I\\\alpha\text{ satisfies }C}}w(C)
\qquad\text{and}\qquad
\UnsatVal_I(\alpha)=\sum_{\substack{C\in I\\\alpha\text{ violates }C}}w(C).
\]
Thus $\Val_I(\alpha)$ is the weight of the satisfied constraints and $\UnsatVal_I(\alpha)$ is the weight of the violated constraints. We define
\[
\OPTmax(I)=\max_\alpha\Val_I(\alpha)
\qquad\text{and}\qquad
\OPTmin(I)=\min_\alpha\UnsatVal_I(\alpha).
\]
The set of satisfying assignments is $\Sat(I)=\{\alpha:\UnsatVal_I(\alpha)=0\}$.

The problem $\SAT(R)$ asks whether $\Sat(I)$ is nonempty. The problem $\MinCSP(R)$ asks for $\OPTmin(I)$, and a $\gamma$-approximation must return a number between $\OPTmin(I)$ and $\gamma\OPTmin(I)$. We write $\MinCSP_\gamma(R)$ for this approximation problem. Finally, $\ExactCSP(R)$ asks for the exact value $\OPTmax(I)$. Instances are unweighted unless we explicitly say otherwise.

We will also use a stronger guarantee than merely approximating $\OPTmin(I)$. A weighted subinstance $H\subseteq I$ is a $(1\pm\varepsilon)$ \emph{sparsifier} of $I$ if $(1-\varepsilon)\UnsatVal_I(\alpha)\le\UnsatVal_H(\alpha)\le(1+\varepsilon)\UnsatVal_I(\alpha)$ for every assignment $\alpha$. As we show in \Cref{cor:mincsp-main}, running the sparsifier with a slightly smaller error and minimizing these inequalities gives a $(1+\varepsilon)$-approximation to $\OPTmin(I)$.

In the streaming model, the constraints arrive one at a time in an arbitrary order, and the algorithm makes one pass over this insertion-only stream. The domain and relation are fixed and are not part of the stream. We use $n$ for the number of variables and $m$ for the stream length. There is no restriction on the update time or on the time used after the stream. We measure only the number of stored bits.

Unless stated otherwise, a randomized algorithm must be correct with probability at least $2/3$ on every input. We write $\SPACE(\Pi)$ for the minimum worst-case space of such a one-pass algorithm for a problem $\Pi$. Upper bounds described as deterministic use no randomness, while all lower bounds also apply to randomized bounded-error algorithms.

\subsection{Non-redundancy}
\label{subsec:nrd}

Non-redundancy measures how many constraints can be independently necessary for satisfiability. Let $I$ be an unweighted $R$-instance and let $C\in I$. We call $C$ \emph{redundant in $I$} if every assignment satisfying all constraints obtained by removing this occurrence of $C$ also satisfies $C$. Otherwise, there is an assignment that violates $C$ and satisfies every other constraint of $I$. We call such an assignment a \emph{witness for $C$}. The instance $I$ is \emph{nonredundant} if every one of its constraints has a witness.

We define $\NRD_n(R)$ to be the maximum number of constraints in a nonredundant $R$-instance on at most $n$ variables.

\begin{example}
For the relation $\CUT$, every forest gives a nonredundant instance. Fix an edge and remove it. The two endpoints are now in different components. We can properly $2$-color each component and choose the colors so that these endpoints receive the same value. This assignment violates the removed edge and satisfies every other edge, so it is a witness for that edge.
\end{example}

We will use the following elementary lower bound to account for the space needed to store an assignment.

\begin{observation}
\label{obs:nrd-linear}
If $R\subsetneq D^k$ is nonempty, then $\NRD_n(R)=\Omega_R(n)$.
\end{observation}

\begin{proof}
Choose a tuple $r\in R$ and a tuple $u\in D^k\setminus R$. Partition $k\lfloor n/k\rfloor$ variables into disjoint ordered scopes of size $k$, and put one $R$-constraint on each scope. To witness one chosen constraint, assign the tuple $u$ to its scope and the tuple $r$ to every other scope. The chosen constraint is violated and all other constraints are satisfied. The resulting instance is nonredundant and has $\lfloor n/k\rfloor=\Omega_R(n)$ constraints.
\end{proof}

By definition, a non-trivial relation is nonempty and is not equal to $D^k$. Thus the observation shows that non-redundancy is at least linear for every non-trivial fixed relation, although it may be much larger. Our satisfiability and minimization results depend on its value for the core of $R$, which we define next.

\subsection{Cores, automorphisms, and hard pinnings}
\label{subsec:core-prelims}

We now define the core of a relation. A map $\sigma:D\to D$ is an \emph{endomorphism} of $R$ if applying $\sigma$ to every coordinate of a tuple in $R$ always gives another tuple in $R$. Equivalently, $\sigma(R)\subseteq R$. An endomorphism is an \emph{automorphism} if it is a bijection. In this case $\sigma(R)=R$. We write $\Aut(R)$ for the set of automorphisms of $R$.

An endomorphism $\sigma$ is \emph{idempotent} if $\sigma(\sigma(a))=\sigma(a)$ for every $a\in D$. Choose an endomorphism with minimum image. Its restriction to its image is a permutation, since applying the map again cannot make the image smaller. A suitable power is therefore the identity on its image and is idempotent. We fix such an idempotent endomorphism $\sigma$ and write $C=\sigma(D)$. We define
\[
\Core(R)=\sigma(R)=R\cap C^k,
\]
and call $C$ its \emph{core domain}. We fix one such choice throughout the paper. The equality holds because an idempotent map fixes every value in its image. The relation $\Core(R)$ is a core relation because every one of its endomorphisms is an automorphism. Indeed, a non-bijective endomorphism composed with $\sigma$ would give an endomorphism of $R$ with image smaller than $C$. We call $R$ itself a \emph{core relation} when every endomorphism of $R$ is an automorphism.

\begin{example}
For $\CUT$, every endomorphism must send $0$ and $1$ to different values. Hence its only endomorphisms are the identity and the map that swaps $0$ and $1$. Both are automorphisms, so $\CUT$ is a core relation.

For a relation that is not a core, let $D=\{0,1,2\}$ and let
\[
R_{\mathrm{star}}=\{(0,1),(1,0),(0,2),(2,0)\}.
\]
The map that fixes $0$ and $1$ and sends $2$ to $1$ is a non-bijective endomorphism of $R_{\mathrm{star}}$. Therefore $R_{\mathrm{star}}$ is not a core relation. Its image is $\{0,1\}$, and restricting the relation to this image gives $\CUT$. Hence $\Core(R_{\mathrm{star}})=\CUT$ on the core domain $\{0,1\}$.
\end{example}

The automorphisms of a core relation are the only unavoidable ambiguity when we try to force variables to chosen values. Let $X$ be a set of variables and let $s:X\to D$ be a fixed assignment. We say that an $R$-instance $P_s$ \emph{hard-pins $s$ up to one common automorphism} if it may use additional variables and
\[
\{\alpha|_X:\alpha\in\Sat(P_s)\}=\{\sigma\circ s:\sigma\in\Aut(R)\}.
\]
\begin{example}
For $\CUT$, introduce two variables $y_0,y_1$ and require $\CUT(y_0,y_1)$. If $s(x)=0$, add $\CUT(x,y_1)$. If $s(x)=1$, add $\CUT(x,y_0)$. These constraints force every $x\in X$ to agree with $s$ after either the identity or the common swap of $0$ and $1$.

The relation $R_{\mathrm{star}}$ above explains why we first reduce to a core. Its non-bijective endomorphism replaces the value $2$ by $1$. Applying this map to any satisfying assignment gives another satisfying assignment, so constraints over $R_{\mathrm{star}}$ cannot in general distinguish these two values. After reducing $R_{\mathrm{star}}$ to its core $\CUT$, the construction above applies.
\end{example}

In \Cref{subsec:hard-pinning}, we prove that every core relation admits the same kind of construction.

\subsection{Finite-domain degree}
\label{subsec:degree}

The parameter for $\ExactCSP$ is the degree of the indicator function of the core relation. We define degree for any function $f:D^k\to\mathbb R$. The domain $D$ need not contain a value already called $0$. We choose any value of $D$ and call it $0$. This choice only fixes the notation, and the reader may think of it as the usual value $0$. For each coordinate $i\in[k]$ and each $a\in D\setminus\{0\}$, define $Z_{i,a}(z)=\mathbf 1[z_i=a]$.

Every function $f:D^k\to\mathbb R$ has a unique expansion
\[
f(z)=\sum_{T\subseteq[k]}\ \sum_{b\in(D\setminus\{0\})^T}\widehat f(T,b)\prod_{i\in T}Z_{i,b_i}(z),
\]
where $b$ chooses one nonzero value $b_i$ for every $i\in T$. These products form a basis for all real-valued functions on $D^k$. This follows by evaluating the function on points that are $0$ outside $T$ and using inclusion-exclusion. In particular, when $f$ is integer-valued, all the coefficients $\widehat f(T,b)$ are integers. We define $\Deg_D(f)$ to be the largest $|T|$ for which $\widehat f(T,b)\ne0$, with constant functions having degree $0$. Changing the basepoint changes the coefficients but not the degree. Indeed, the indicator of the old basepoint is one minus the sum of the indicators of all other values. Thus either set of indicators can be written using the other without increasing degree.

For a relation $R\subseteq D^k$, we write $\mathbf 1_R$ for its indicator function and use $\Deg_D(\mathbf 1_R)$ as the degree of $R$.

\begin{example}
For $\CUT$, let $Z_i(z)=\mathbf 1[z_i=1]$. Then
\[
\mathbf 1_{\CUT}(z_1,z_2)=Z_1(z)+Z_2(z)-2Z_1(z)Z_2(z),
\]
so $\Deg_{\{0,1\}}(\mathbf 1_{\CUT})=2$.

As another example, let $R_{\NAE}\subseteq\{0,1\}^3$ contain the tuples that are not constant. With $Z_i(z)=\mathbf 1[z_i=1]$, we have
\[
\mathbf 1_{R_{\NAE}}(z)=Z_1(z)+Z_2(z)+Z_3(z)-Z_1(z)Z_2(z)-Z_1(z)Z_3(z)-Z_2(z)Z_3(z).
\]
Thus $\Deg_{\{0,1\}}(\mathbf 1_{R_{\NAE}})=2$, even though $R_{\NAE}$ has arity $3$.
\end{example}

The upper and lower bounds for $\ExactCSP$ will both use this degree after reducing the relation to its core.

\subsection{Communication complexity}
\label{subsec:communication}

Our streaming lower bounds use the one-way communication problem $\Index_M$. Alice receives a string $b\in\{0,1\}^M$, Bob receives an index $i\in[M]$, and Bob must output $b_i$. In a public-coin one-way protocol, Alice and Bob share a random string, Alice sends one message to Bob, and Bob produces the answer from this message, his index, and the shared randomness.

The connection to streaming is standard. Alice creates the first part of a stream from her input, runs the streaming algorithm on it, and sends the memory state to Bob. Bob resumes the algorithm from this state and appends constraints determined by his index. The memory state is the only message, so an $s$-space streaming algorithm gives a one-way protocol with $s$ bits of communication. The random tape of the streaming algorithm can be treated as the public random string.

We use the following standard lower bound.

\begin{lemma}[One-way complexity of $\Index$, {\cite{Index1,Index2}}]
\label{lem:index}
Every randomized public-coin one-way protocol for $\Index_M$ with constant error uses $\Omega(M)$ bits of communication.
\end{lemma}

In each of our lower bounds, Alice's part of the stream encodes her string, while Bob appends constraints that reveal the bit at his index. The rest of each reduction is devoted to constructing those constraints.

\section{Properties of Cores}
\label{sec:cores}

In this section, we prove the facts about cores that we will use throughout the paper. We first show that reducing to the core does not change any of our three objectives. We then give two ways to force an assignment using only constraints from a core relation. The first forces the assignment up to one common automorphism and will be used for satisfiability. The second makes the same idea strong enough for exact optimization.

We will repeatedly use the fact that applying one automorphism to every assigned value changes nothing. We state it in the following observation.

\begin{observation}
\label{obs:automorphism-value}
Let $R\subseteq D^k$, let $I$ be an $R$-instance, and let $\sigma\in\Aut(R)$. For every assignment $\alpha$, we have $\Val_I(\sigma\circ\alpha)=\Val_I(\alpha)$. Therefore, $\UnsatVal_I(\sigma\circ\alpha)=\UnsatVal_I(\alpha)$.
\end{observation}

\begin{proof}
Consider a constraint $R(x_1,\ldots,x_k)$ of $I$. Since $\sigma$ is an automorphism, the constraint is satisfied by $\alpha$ if and only if it is satisfied by $\sigma\circ\alpha$. This holds for every constraint, which proves both equalities.
\end{proof}

\subsection{Reduction to the core}
\label{subsec:core-reduction}

The first goal is to reduce all three streaming problems from $R$ to its core $R^\star$. The reduction keeps every constraint on the same scope. It only replaces the relation $R$ by $R^\star$. The next lemma states that this preserves the streaming space exactly.

\begin{lemma}
\label{lem:core-space}
Let $R^\star=\Core(R)$. The following statements hold.
\begin{enumerate}[label=(\roman*),leftmargin=*]
\item $\SPACE(\SAT(R))=\SPACE(\SAT(R^\star))$.
\item for every $\gamma\ge 1$,
$\SPACE(\MinCSP_\gamma(R))=\SPACE(\MinCSP_\gamma(R^\star))$.
\item $\SPACE(\ExactCSP(R))=\SPACE(\ExactCSP(R^\star))$.
\end{enumerate}
The reductions preserve the number of variables and constraints, the scope of each corresponding constraint, stream order, and simplicity. If $R$ is symmetric, then $R^\star$ is symmetric as well.
\end{lemma}

To prove \Cref{lem:core-space}, we first prove that the same replacement preserves the two optimum values of every instance.

\begin{lemma}
\label{lem:core-reduction}
Let $R\subseteq D^k$, let $\sigma:D\to C$ be an idempotent endomorphism of minimum image, and let $R^\star=\Core(R)=\sigma(R)=R\cap C^k$. Given an $R$-instance $I$, let $I^\star$ be the $R^\star$-instance obtained by replacing every constraint $R(x_1,\ldots,x_k)$ by $R^\star(x_1,\ldots,x_k)$ on the same scope. Then $\OPTmax(I)=\OPTmax(I^\star)$ and $\OPTmin(I)=\OPTmin(I^\star)$. In particular, $I$ is satisfiable if and only if $I^\star$ is satisfiable.
\end{lemma}

\begin{proof}
Let $\alpha$ be any assignment for $I$. If $R(x_1,\ldots,x_k)$ is satisfied by $\alpha$, then $(\alpha(x_1),\ldots,\alpha(x_k))\in R$, and applying the endomorphism $\sigma$ gives $(\sigma(\alpha(x_1)),\ldots,\sigma(\alpha(x_k)))\in\sigma(R)=R^\star$. Hence $\Val_{I^\star}(\sigma\circ\alpha)\ge\Val_I(\alpha)$ and $\UnsatVal_{I^\star}(\sigma\circ\alpha)\le\UnsatVal_I(\alpha)$. Taking an assignment that is optimal for the relevant objective gives $\OPTmax(I^\star)\ge\OPTmax(I)$ and $\OPTmin(I^\star)\le\OPTmin(I)$.

For the reverse direction, consider any assignment $\beta$ taking values in $C$. Since $R^\star=R\cap C^k$, every constraint is satisfied by $\beta$ as an $R$-constraint if and only if it is satisfied by $\beta$ as an $R^\star$-constraint. In particular, an assignment attaining either optimum for $I^\star$ is also an assignment for $I$ with the same value. Therefore $\OPTmax(I)\ge\OPTmax(I^\star)$ and $\OPTmin(I)\le\OPTmin(I^\star)$. Combining the two directions proves both equalities, and the satisfiability claim follows from $\OPTmin=0$.
\end{proof}

We can now prove the streaming reduction. The point is that the instance level argument did not change any scope.

\begin{proof}[Proof of \Cref{lem:core-space}.]
Given a stream of $R$-constraints, replace each arriving constraint by the $R^\star$-constraint on the same scope and run an algorithm for $R^\star$. By \Cref{lem:core-reduction}, the final instance has the same two optimum values and is satisfiable exactly when the original instance is satisfiable. Conversely, replace every $R^\star$-constraint by the $R$-constraint on the same scope and use the same lemma again. Both transformations use only constant extra space and preserve the number of variables and constraints, the stream order, and simplicity. Finally, if $R$ is symmetric, then $R^\star=R\cap C^k$ is symmetric as well.
\end{proof}

Thus, from this point on, we may reduce to the core before proving any upper or lower bound.

\subsection{Hard pinnings in a core relation}
\label{subsec:hard-pinning}

For the satisfiability lower bound, Bob will want to force the variables to a chosen witness assignment $s$. Ordinary $R$-constraints may not be able to distinguish $s$ from $\sigma\circ s$ for an automorphism $\sigma$, but this ambiguity is harmless by \Cref{obs:automorphism-value}. The next lemma shows that this is the only ambiguity we need to keep.

\begin{lemma}[Hard pinning]
\label{lem:hard-pinning}
Let $R\subseteq D^k$ be a core relation, let $X$ be a set of variables, and let $s:X\to D$. There is an $R$-instance $P_s$ on $X$ and $O_R(1)$ additional variables such that
\[
\{\alpha|_X:\alpha\in\Sat(P_s)\}=\{\sigma\circ s:\sigma\in\Aut(R)\}.
\]
Moreover, $P_s$ has $O_R(|X|+1)$ constraints.
\end{lemma}

The proof uses the following property of cores. Changing one domain value to another must make some tuple leave the relation.

\begin{lemma}
\label{lem:core-separates-values}
Let $R\subseteq D^k$ be a core relation. For every distinct $a,b\in D$, there is a tuple $r=(r_1,\ldots,r_k)\in R$ and a coordinate $t\in[k]$ such that $r_t=a$ and $r^{t\leftarrow b}\notin R$.
\end{lemma}

\begin{proof}
Define $\sigma:D\to D$ by $\sigma(a)=b$ and $\sigma(c)=c$ for $c\ne a$. The map $\sigma$ is not injective, so it cannot be an endomorphism of the core relation $R$. Choose $r\in R$ such that $\sigma(r)\notin R$, and replace the occurrences of $a$ in $r$ by $b$ one at a time. The initial tuple is in $R$ and the final tuple is not, so at some first replacement we obtain a tuple outside $R$. Immediately before that replacement we have a tuple $r'\in R$ with $r'_t=a$, while $(r')^{t\leftarrow b}\notin R$. This is the required tuple.
\end{proof}

We now build the hard-pinning instance. Intuitively, we introduce additional variables to represent the values of $D$, and a group of constraints forces their assigned values to form one automorphism of $R$.

\begin{proof}[Proof of \Cref{lem:hard-pinning}.]
The first step is to introduce extra variables that represent the values in the domain, up to one common automorphism. For $\CUT$, these are two adjacent variables, so they must receive opposite values. More formally, for every $a\in D$, introduce an extra variable $y_a$. For every tuple $r=(r_1,\ldots,r_k)\in R$, add the constraint
$R(y_{r_1},\ldots,y_{r_k}).$
Next, we add constraints that pin the variables in $X$ to $s$ up to the same automorphism. For every $x\in X$, every $r\in R$, and every $t\in[k]$ with $r_t=s(x)$, also add $R(y_{r_1},\ldots,y_{r_{t-1}},x,y_{r_{t+1}},\ldots,y_{r_k})$.
Call the resulting instance $P_s$. It uses $|D|=O_R(1)$ extra variables and $O_R(|X|+1)$ constraints.

We now prove that the added constraints give the required hard pinning. Let $\alpha$ satisfy $P_s$ and define $\sigma:D\to D$ by $\sigma(a)=\alpha(y_a)$. The first family of constraints shows that $\sigma(r)\in R$ for every $r\in R$, so $\sigma$ is an endomorphism of $R$. Since $R$ is a core, $\sigma$ is an automorphism. Fix $x\in X$, put $a=s(x)$, and let $b=\sigma^{-1}(\alpha(x))$. Suppose $b\ne a$. By \Cref{lem:core-separates-values}, there are $r\in R$ and $t\in[k]$ such that $r_t=a$ and $r^{t\leftarrow b}\notin R$. The corresponding constraint from the second family is assigned the tuple $\sigma(r^{t\leftarrow b})$, which is not in $R$ because $\sigma$ is an automorphism. This contradicts that $\alpha$ satisfies $P_s$. Hence $\alpha(x)=\sigma(s(x))$ for every $x\in X$.

Conversely, fix any $\sigma\in\Aut(R)$ and assign $y_a\mapsto\sigma(a)$ and $x\mapsto\sigma(s(x))$. Every constraint in both families is then assigned a tuple of the form $\sigma(r)$ for some $r\in R$, so all constraints are satisfied. The restrictions of satisfying assignments to $X$ are therefore exactly the assignments $\sigma\circ s$ with $\sigma\in\Aut(R)$.
\end{proof}

The gadget therefore behaves like a genuine pinning whenever the objective is unchanged by a common automorphism. For exact optimization we need more because the goal is not merely to satisfy the whole instance. An assignment may gain constraints in the original instance by violating some constraints in the pinning gadget. We therefore need to ensure that deviating from the intended assignment loses more gadget constraints than it can gain in the original instance. We do this next.

\subsection{Pinning for exact optimization}
\label{subsec:exact-pinning}

We now strengthen the construction for exact optimization. The input instance $P$ may reward an assignment for moving away from $s$, so a merely satisfiable pinning gadget is not enough. Intuitively, we replace each domain value by a block of many variables. The constraints among the blocks create a large loss unless their most frequent values form an automorphism, and the constraints joining $P$ to the blocks create a large loss whenever an original variable disagrees with $s$ under that automorphism.

\begin{lemma}
\label{lem:exact-pinning}
Let $R\subseteq D^k$ be a core relation with $k\ge 2$, and fix $1\le d\le k$. Let $P$ be a simple $R$-instance on $O_R(N)$ variables, with at most $N^d$ constraints and maximum variable degree at most $N^{d-1}$. For every assignment $s$ to the variables of $P$, there is an unweighted $R$-instance $H_s$, using the variables of $P$ and $O_R(N)$ additional variables, such that $P\cup H_s$ is simple and
\[
\OPTmax(P\cup H_s)=|H_s|+\Val_P(s).
\]
If $R$ is symmetric and $P$ has no repeated unordered scope, then $P\cup H_s$ has no repeated unordered scope.
Moreover, $H_s$ can be constructed from $R$, the variable set of $P$, the assignment $s$, and the public parameters $N$ and $d$. It does not depend on the constraints of $P$.
\end{lemma}

\begin{proof}[Proof of \Cref{lem:exact-pinning}.]
Let $X$ be the variable set of $P$, let $B$ be the number of constraints in $P$, and let $\Delta$ be its maximum variable degree. Thus $B\le N^d$ and $\Delta\le N^{d-1}$. We will choose a constant $C_R$ below and set $m=C_RN$.

The construction follows the same idea as \Cref{lem:hard-pinning}. We first create many variables that represent each value of the domain up to one common automorphism. We then connect the variables of $P$ to these new variables. More formally, for every $a\in D$, create a \emph{block} $Z_a$ of $m$ new variables. For every tuple $r=(r_1,\ldots,r_k)\in R$ and every choice of pairwise distinct variables $z_i\in Z_{r_i}$, add the constraint $R(z_1,\ldots,z_k)$. We call these the \emph{block constraints}. Next fix $x\in X$ and put $a=s(x)$. For every $r\in R$, every $t\in[k]$ with $r_t=a$, and every choice of pairwise distinct variables $z_j\in Z_{r_j}$ for $j\ne t$, add the constraint $R(z_1,\ldots,z_{t-1},x,z_{t+1},\ldots,z_k)$. We call these the \emph{attached constraints for $x$}. If the same ordered scope arises more than once, keep only one constraint. When $R$ is symmetric, keep only one constraint for each unordered scope. Let $H_s$ be the union of the block constraints and all the attached constraints.

Firstly note that the instance $P\cup H_s$ is simple. This is because the constraints of $P$ use only variables in $X$, block constraints use only variables from the blocks, and an attached constraint for $x$ uses $x$ together with $k-1$ variables from the blocks. The removal of duplicate scopes makes each family simple, and attached constraints for two different variables of $X$ cannot have the same scope. The number of new variables is $|D|m=O_R(N)$. This construction uses the variable set of $P$ but not its constraints, as required in the final part of the lemma.

We next show the counting bound that we need. Let $q=|D|$. Under any assignment $\alpha$, choose a most frequent value in each block $Z_a$ and denote it by $\sigma(a)$. This value appears at least $m/q$ times in $Z_a$. Consider any sequence of $\ell\le k$ blocks and choose one variable from each block in order. At every step, at most $k-1$ variables are forbidden by the earlier choices. Hence, when $m\ge2qk$, there are at least $m/(2q)$ choices at every step and at least $(m/(2q))^\ell$ pairwise distinct choices in total. At most $k!$ ordered choices give the same unordered scope. We may therefore fix a constant $\eta_R>0$ such that there are at least $\eta_Rm^\ell$ distinct ordered scopes, and also at least $\eta_Rm^\ell$ distinct unordered scopes when $R$ is symmetric. For a fixed tuple $r$, different ordered choices give different ordered scopes of block constraints. The same is true for the attached constraints when $x$, $r$, and $t$ are fixed. Thus the bound continues to hold after we remove repeated scopes.

Choose $C_R$ large enough that $C_R\ge2qk$, $\eta_RC_R^k>1$, and $\eta_RC_R^{k-1}>1$. Since $m=C_RN$, $B\le N^d$, $\Delta\le N^{d-1}$, and $d\le k$, we then have $\eta_Rm^k>B$ and $\eta_Rm^{k-1}>\Delta$ for every $N\ge1$.

We can now compare every assignment with the intended one. Assigning $x\mapsto s(x)$ for $x\in X$ and assigning every variable in $Z_a$ the value $a$ satisfies all constraints of $H_s$. Therefore $\OPTmax(P\cup H_s)\ge |H_s|+\Val_P(s)$.

For the reverse inequality, fix an arbitrary assignment $\alpha$ and define $\sigma:D\to D$ from the most frequent values of the blocks as above. Suppose first that $\sigma$ is not an endomorphism of $R$. Choose $r\in R$ with $\sigma(r)\notin R$. At least $\eta_Rm^k>B$ block constraints are assigned the tuple $\sigma(r)$, so all of them are violated. Even if $\alpha$ satisfies every constraint of $P$, it has value less than $|H_s|$, while the intended assignment has value at least $|H_s|$. Thus any assignment that attains the optimum must have $\sigma$ as an endomorphism and, since $R$ is a core, as an automorphism.

Thus, we only need to consider the case in which $\sigma$ is an automorphism. Let $W=\{x\in X:\alpha(x)\ne\sigma(s(x))\}$. Fix $x\in W$, put $a=s(x)$, and put $b=\sigma^{-1}(\alpha(x))$, so $b\ne a$. By \Cref{lem:core-separates-values}, there are $r\in R$ and $t\in[k]$ such that $r_t=a$ and $r^{t\leftarrow b}\notin R$. In every other coordinate $j$, choose a variable from $Z_{r_j}$ assigned its most frequent value $\sigma(r_j)$. There are at least $\eta_Rm^{k-1}$ such attached constraints for $x$, and each is assigned the tuple $\sigma(r^{t\leftarrow b})$. All of them are violated because $\sigma$ is an automorphism. The constraints counted for different variables $x$ are disjoint, so $\alpha$ violates at least $\eta_Rm^{k-1}|W|$ attached constraints.

The assignment $\sigma\circ s$ has value $\Val_P(s)$ on $P$ by \Cref{obs:automorphism-value}. Changing the variables in $W$ can affect at most $\Delta|W|$ constraints of $P$, and hence $\Val_P(\alpha|_X)\le\Val_P(s)+\Delta|W|$. Since $\alpha$ violates at least $\eta_Rm^{k-1}|W|$ attached constraints, we have $\Val_{P\cup H_s}(\alpha)\le |H_s|+\Val_P(s)+(\Delta-\eta_Rm^{k-1})|W|\le |H_s|+\Val_P(s)$. Together with the intended assignment, this proves the required equality.
\end{proof}

The exact-pinning gadget still leaves one common automorphism, but that is fine since the objective of $P$ is invariant under it, while every other disagreement with $s$ loses more gadget constraints than it can gain in $P$.

\section{Satisfiability}
\label{sec:sat}

We now characterize the one-pass streaming space needed to decide satisfiability. We will prove that after reducing to the core, the right parameter is non-redundancy. The upper bound stores a small instance with the same satisfying assignments as the constraints seen so far. For the lower bound, a maximum nonredundant instance gives the bits of an $\Index$ instance, and hard pinning lets Bob test the bit corresponding to his constraint.

\begin{theorem}[Satisfiability]
\label{thm:sat-main}
Let $R\subseteq D^k$ be a fixed non-trivial relation and let $R^\star=\Core(R)$. The following statements hold.
\begin{enumerate}[label=(\roman*),leftmargin=*]
\item there is a deterministic one-pass streaming algorithm for $\SAT(R)$ using $O_R(\NRD_n(R^\star)\log n)$ bits of space.
\item every randomized bounded-error one-pass streaming algorithm for $\SAT(R)$ uses $\Omega_R(\NRD_n(R^\star))$ bits of space.
\end{enumerate}
Consequently,
\[
\Omega_R\!\left(\NRD_n(\Core(R))\right)
\le \SPACE(\SAT(R))
\le O_R\!\left(\NRD_n(\Core(R))\log n\right).
\]
\end{theorem}

We prove the two bounds for a core relation and then use \Cref{lem:core-space} to obtain the theorem for an arbitrary relation.

The upper bound is the following known consequence of non-redundancy \cite[Theorem~1(a)]{sharma2026characterizingstreamingdecidabilitycsps}.

\begin{lemma}[Satisfiability upper bound]\cite[Theorem~1(a)]{sharma2026characterizingstreamingdecidabilitycsps}
\label{lem:sat-upper}
If $R$ is a core relation, then $\SAT(R)$ has a deterministic one-pass streaming algorithm using $O_R(\NRD_n(R)\log n)$ bits of space.
\end{lemma}

The algorithm maintains a nonredundant subinstance with exactly the same satisfying assignments as the current prefix of the stream. When a constraint arrives, the algorithm adds it and repeatedly removes any constraint whose deletion does not change the set of satisfying assignments. The stored instance therefore has at most $\NRD_n(R)$ constraints. Since $R$ is fixed, each constraint is specified by its scope using $O_R(\log n)$ bits. There is no update-time restriction, so redundancy and final satisfiability can be checked by trying all assignments using $O_R(n)$ additional bits. For a non-trivial relation, $\NRD_n(R)=\Omega_R(n)$, so this is within the stated space bound. For a trivial relation, satisfiability is immediate.

It remains to prove the matching lower bound. The reduction uses the witnesses of a nonredundant instance in the way suggested by the definition of non-redundancy.

\subsection{Lower bound for a core relation}
\label{subsec:sat-lb}

In this subsection, we prove the following main theorem.

\begin{theorem}[Satisfiability lower bound]
\label{thm:sat-core-lb}
Let $R\subseteq D^k$ be a core relation. Every randomized bounded-error one-pass streaming algorithm for $\SAT(R)$ uses $\Omega_R(\NRD_n(R))$ bits of space.
\end{theorem}

\begin{proof}
The idea is simple. Alice streams a subinstance of a maximum nonredundant instance. Bob takes the witness for his queried constraint, meaning an assignment that violates this constraint and satisfies every other constraint in the nonredundant instance, and uses \Cref{lem:hard-pinning} to force that witness up to one automorphism. Since automorphisms preserve which constraints are satisfied, the resulting instance tells Bob whether Alice streamed the queried constraint.

Let $c_R$ be a constant such that the hard-pinning instance in \Cref{lem:hard-pinning} always uses at most $c_R$ additional variables. Fix a sufficiently large number $n$, put $n_0=n-c_R$, and let $F=\{C_1,\ldots,C_M\}$ be a nonredundant $R$-instance on $n_0$ variables with $M=\NRD_{n_0}(R)$. If $M=0$, the claimed lower bound is empty, so assume $M\ge1$. For every $i\in[M]$, choose a witness assignment $s_i$ that violates $C_i$ and satisfies every $C_j$ with $j\ne i$.

Suppose that a randomized one-pass streaming algorithm decides $\SAT(R)$ on $n$ variables using $S(n)$ bits of space. We use it to solve $\Index_M$. Alice receives a string $x\in\{0,1\}^M$, and Bob receives an index $i\in[M]$. Alice starts the streaming algorithm and streams $C_j$ exactly when $x_j=1$. She then sends the $S(n)$-bit memory state to Bob. Bob resumes the algorithm from this state and appends the hard-pinning instance $P_{s_i}$ from \Cref{lem:hard-pinning}. The final instance uses at most $n_0+c_R=n$ variables.

We claim that the final instance is satisfiable if and only if $x_i=0$. Suppose first that $x_i=0$. The assignment $s_i$ satisfies every constraint streamed by Alice, and it extends to a satisfying assignment of $P_{s_i}$ by using the identity automorphism in \Cref{lem:hard-pinning}. Hence the final instance is satisfiable.

Now suppose that $x_i=1$. Any assignment satisfying $P_{s_i}$ restricts on the original variables to $\sigma\circ s_i$ for some $\sigma\in\Aut(R)$. The assignment $s_i$ violates $C_i$, so \Cref{obs:automorphism-value} implies that $\sigma\circ s_i$ also violates $C_i$. Since Alice streamed $C_i$, no assignment can satisfy the final instance. This proves the claim.

Alice and Bob use the random tape of the streaming algorithm as public randomness. Bob outputs the complement of the algorithm's answer, which recovers $x_i$ with the same constant error probability. We have therefore obtained a public-coin one-way protocol for $\Index_M$ with $S(n)$ bits of communication. By \Cref{lem:index}, $S(n)=\Omega(M)=\Omega(\NRD_{n_0}(R))$.

Finally, we need to prove that removing a constant number of variables changes non-redundancy by at most a constant factor depending only on $R$. We start with a maximum nonredundant instance on at most $n$ variables and add unused variables so that its variable set has size exactly $n$. Choose a uniformly random set of $n_0=n-c_R$ variables. A constraint uses at most $k$ distinct variables, so for all sufficiently large $n$ the probability that all variables in its scope are chosen is bounded below by a positive constant dependent only on $R$. Some choice therefore keeps $\Omega_R(\NRD_n(R))$ constraints. These constraints remain nonredundant because their original witnesses can be restricted to the chosen variables. Thus $\NRD_{n_0}(R)=\Omega_R(\NRD_n(R))$, and consequently $S(n)=\Omega_R(\NRD_n(R))$.
\end{proof}

We now combine the core reduction with the two core relation bounds to conclude the proof of \Cref{thm:sat-main}.

\begin{proof}[Proof of \Cref{thm:sat-main}.]
By \Cref{lem:core-space}, deciding $\SAT(R)$ has the same streaming space complexity as deciding $\SAT(R^\star)$. Apply \Cref{lem:sat-upper} to $R^\star$ for the deterministic upper bound and \Cref{thm:sat-core-lb} to $R^\star$ for the randomized lower bound. This gives both parts of the theorem and the claimed characterization.
\end{proof}

Thus the streaming complexity of satisfiability is governed, up to the logarithmic factor in the upper bound, by the non-redundancy of the core relation.

\section{Sparsification and \texorpdfstring{\MinCSP}{Min-CSP}}
\label{sec:mincsp}

We first state our sparsification result for \CSPs that are revealed in streams.

\begin{theorem}[Sparsification]
\label{thm:streaming-sparsification}
Let $R\subseteq D^k$ be a fixed non-trivial relation and let $\varepsilon>0$. There is a deterministic one-pass streaming algorithm which, given an unweighted stream of constraints defining an $R$-instance $I$ on $n$ variables and with stream length $m$, stores a weighted subinstance $H\subseteq I$ such that for every assignment $\alpha$,
$
\UnsatVal_H(\alpha)\in(1\pm\varepsilon)\UnsatVal_I(\alpha).
$
The space used is $O_{R,\varepsilon}(\NRD_n(R)\cdot\operatorname{polylog}(n,m,1/\varepsilon))$.
\end{theorem}

As an immediate corollary, we obtain the following.

\begin{corollary}[\MinCSP]
\label{cor:mincsp-main}
Let $R\subseteq D^k$ be a fixed non-trivial relation, let $R^\star=\Core(R)$, and let $\varepsilon>0$.
\begin{enumerate}[label=(\roman*),leftmargin=*]
\item there is a deterministic one-pass $(1+\varepsilon)$-approximation algorithm for $\MinCSP(R)$ using
\(
O_{R,\varepsilon}\!\left(\NRD_n(R^\star)
\operatorname{polylog}(n,m,1/\varepsilon)\right)
\)
bits of space;
\item every randomized bounded-error one-pass finite-factor approximation algorithm for $\MinCSP(R)$ uses $\Omega_R(\NRD_n(R^\star))$ bits.

\end{enumerate}
Consequently,
\[
\begin{aligned}
\Omega_R\!\left(\NRD_n(\Core(R))\right)
\le \SPACE(\MinCSP_{1+\varepsilon}(R)) 
\le O_{R,\varepsilon}\!\left(\NRD_n(\Core(R))
\operatorname{polylog}(n,m,1/\varepsilon)\right).
\end{aligned}
\]
\end{corollary}

We first prove \Cref{thm:streaming-sparsification}.  The proof of \Cref{cor:mincsp-main} is given at the end of the section.

\subsection{Sparsification}
\label{subsec:sparsification}

At a high level, we use a sparsify--merge--sparsify procedure for building sparsifiers in insertion-only streams. These sparsifiers have been used before in insertion-only streams; see, for instance, \cite{mcgregor2014graph}. The only sparsification result we need as a black box is the following theorem of
Brakensiek and Guruswami.
\begin{fact}[\cite{brakensiek2025redundancy}]
    If
$\mathcal C\subseteq\{0,1\}^{M}$ is a binary code and $0<\eta<1$, then
there is a nonnegative weight vector $z$ such that
\[
  \langle z,c\rangle\in(1\pm\eta)\lVert c\rVert_1
  \qquad\text{for every }c\in\mathcal C,
\]
and, for an absolute constant $c_{\mathrm{BG}}>0$,
\begin{equation}
\label{eq:bg-code-sparsifier}
  |\operatorname{supp}(z)|
  \leq
  c_{\mathrm{BG}}\,\NRD(\mathcal C)
  \frac{\log_2^6(4M)}{\eta^2}.
\end{equation}
Importantly, their proof allows $z$ to be chosen with nonnegative integer entries. Indeed, their induction starts with unit weights and, in the recursive step, assigns weight $1$ to the explicitly retained coordinates and multiplies the recursively produced weights by $3$.
\end{fact}

Subsequent work sharpened and broadened this code-sparsification framework. Lovett, Meka, and Wang \cite{LovettMW26} improved the dependence on the block length from polylogarithmic to logarithmic, while Khanna, Putterman, and Sudan \cite{KhannaPS26} developed a unified sparsification theory for nonnegative real-valued codes.

Note that in the above, $\NRD(\mathcal C)$ refers to the largest \emph{identity} submatrix that can be found in $\mathcal{C}$ (when its vectors are listed as rows).

We now prove the following basic fact which shows that weighted instances (with small weights) can effectively be sparsified as \emph{unweighted} instances, just by repeating constraints.

\begin{lemma}[Integral-weight reduction]
\label{lem:integral-weight-reduction}
Let $J$ be an integer-weighted $R$-instance on at most $n$ variables. Let
\[
  W(J)=\sum_{C\in\operatorname{supp}(J)}w_J(C),
\]
denote the total weight of the instance. For every $0<\eta<1$, there is an integer-weighted subinstance $\operatorname{Red}_{\eta}(J)\subseteq J$ such that, simultaneously for every assignment $\alpha$,
\begin{equation}
\label{eq:reduction-guarantee}
  \UnsatVal_{\operatorname{Red}_{\eta}(J)}(\alpha)
  \in(1\pm\eta)\UnsatVal_J(\alpha),
\end{equation}
and
\begin{equation}
\label{eq:reduction-support}
  |\operatorname{supp}(\operatorname{Red}_{\eta}(J))|
  \leq
  c_{\mathrm{BG}}\,\NRD_n(R)
  \frac{\log_2^6(4\max\{W(J),1\})}{\eta^2}.
\end{equation}
\end{lemma}

\begin{proof}
The assertion is immediate when $W(J)=0$, so assume $W(J)>0$.  Because the instance is integer weighted, we simply replace each constraint $C$ of weight $w_J(C)$ by $w_J(C)$ labeled unit-weight copies.  Denote the resulting unweighted multiset by $J^{\uparrow}$; it has exactly $W(J)$ occurrences. Now, for this \emph{unweighted} \CSP instance, we construct the canonical ``violation code''.
\begin{equation}
\label{eq:violation-code}
  \mathcal C_{J^{\uparrow}}
  =
  \left\{
    \bigl(\mathbf 1[\alpha\text{ violates }C]\bigr)_{C\in J^{\uparrow}}
    :\alpha\colon V(J)\to D
  \right\}
  \subseteq\{0,1\}^{W(J)}.
\end{equation}
We claim that
\begin{equation}
\label{eq:violation-code-nrd}
  \NRD(\mathcal C_{J^{\uparrow}})
  \leq \NRD_n(R).
\end{equation}
Indeed, suppose that $t$ coordinates form an identity submatrix, and let $C_1,\ldots,C_t$ be the corresponding constraint copies.  For every $i\in[t]$, the row containing the $i$th diagonal $1$ is an assignment
that violates $C_i$ and satisfies every $C_j$ with $j\neq i$.  Hence $\{C_1,\ldots,C_t\}$ is a nonredundant $R$-instance on at most $n$ variables.  Notice also that two copies of the same constraint, or two distinct constraints operating on the same ordered tuple of variables, have identical columns and therefore cannot both appear in such an identity submatrix.

Next, we apply \eqref{eq:bg-code-sparsifier} to $\mathcal C_{J^{\uparrow}}$.  Its integer weight vector approximately preserves the number of violated copies for every assignment.  To build the ultimate sparsifier, we collapse all selected copies belonging to the same occurrence of $J$ and add their weights. Since we impose no restriction on computation time, the required weight vector can be found deterministically within the stated space bound by enumerating candidate supports and integer weights and checking the inequalities over all assignments. Coordinates that are identically zero may be discarded, while every remaining selected weight is at most $(1+\eta)W(J)$, so this search is finite. This then yields \eqref{eq:reduction-guarantee} and \eqref{eq:reduction-support}.
\end{proof}

Now, we proceed to give the streaming construction. It is enough to consider an ultimate sparsifier accuracy of $0<\varepsilon\leq 1/2$. For a larger error, run the construction with error $1/2$. We then set
\begin{equation}
\label{eq:error-schedule}
  \eta_\ell=\frac{\varepsilon}{20(\ell+1)^2}
  \quad(\ell\geq1).
\end{equation}

The algorithm maintains at most one sparsifier at each level $\ell\geq0$.  A level-$\ell$ sparsifier represents a block of exactly $2^\ell$ original constraints in the stream. Whenever the algorithm stores two sparsifiers at level $\ell$, we take their weighted union $J$ and replace it by
\[
  \operatorname{Red}_{\eta_{\ell+1}}(J),
\]
which now constitutes a sparsifier at level $\ell+1$.  We repeat this carry operation until every level again contains at most one sparsifier.

\begin{claim}
\label{claim:merge-reduce-error}
Let $H_B$ be a level-$\ell$ sparsifier and let $I_B$ be the unweighted
subinstance consisting of its underlying block of constraints.  Then, for
every assignment $\alpha$,
\begin{equation}
\label{eq:level-error-invariant}
  \left(\prod_{j=1}^{\ell}(1-\eta_j)\right)
  \UnsatVal_{I_B}(\alpha)
  \leq \UnsatVal_{H_B}(\alpha)
  \leq
  \left(\prod_{j=1}^{\ell}(1+\eta_j)\right)
  \UnsatVal_{I_B}(\alpha).
\end{equation}
Consequently, the union of the sparsifiers remaining at the end of the
stream is a $(1\pm\varepsilon)$-sparsifier of the entire input.
\end{claim}

\begin{proof}
We prove~\eqref{eq:level-error-invariant} by induction on $\ell$.  It is
an equality at level zero.  Suppose two level-$(\ell-1)$ sparsifiers $H_{B_1}$ and $H_{B_2}$ are merged.  Since all values are nonnegative, adding their two inductive inequalities shows that their weighted union has the same multiplicative bounds relative to $I_{B_1}\uplus I_{B_2}$. Applying $\operatorname{Red}_{\eta_\ell}$ multiplies the lower and upper bounds by $1-\eta_\ell$ and $1+\eta_\ell$, respectively, thereby yielding the induction. 

The schedule satisfies
\[
  s:=\sum_{\ell\geq1}\eta_\ell
  <\frac{\varepsilon}{20}\leq\frac1{40}.
\]
Therefore
\[
  \prod_{\ell\geq1}(1-\eta_\ell)
  \geq1-s\geq1-\varepsilon
\]
and, using $e^x\leq1+2x$ for $0\leq x\leq1/2$,
\[
  \prod_{\ell\geq1}(1+\eta_\ell)
  \leq e^s\leq1+2s\leq1+\varepsilon.
\]
This then finishes the claim.
\end{proof}

\begin{claim}
\label{claim:merge-reduce-space}
Let $d=\NRD_n(R)$.  A level-$\ell$ sparsifier has support size at most
\begin{equation}
\label{eq:level-support-bound}
  O_R\!\left(
    \frac{d}{\varepsilon^2}(\ell+1)^{10}
  \right),
\end{equation}
and each of its positive weights is an integer smaller than $2^{\ell+1}$.  On a stream of length $m$, the total space used is bounded by
\begin{equation}
\label{eq:stream-space-bound}
  O_R\!\left(
    \frac{d}{\varepsilon^2}
    \log^{11}(m+2)  \bigl(\log(n+2)+\log(m+2)+\log(1/\varepsilon+2)\bigr)
  \right)
\end{equation}
bits.
\end{claim}

\begin{proof}
Let
\[
  b_\ell=\prod_{j=1}^{\ell}(1+\eta_j).
\]
The proof of Claim~\ref{claim:merge-reduce-error} gives $b_\ell<2$. If a constraint $C$ has positive weight in a level-$\ell$ sparsifier $H_B$, then $C$
is not vacuously true (otherwise, we simply ignore the constraint), so some assignment $\alpha_C$ violates it.  Hence
\begin{equation}
\label{eq:level-weight-bound}
  w_{H_B}(C)
  \leq\UnsatVal_{H_B}(\alpha_C)
  \leq b_\ell\UnsatVal_{I_B}(\alpha_C)
  \leq b_\ell|I_B|
  <2^{\ell+1}.
\end{equation}
Integrality follows from Lemma~\ref{lem:integral-weight-reduction}.

Put $q=|D|$.  A uniformly random assignment violates every occurrence
in the support of $H_B$ with probability at least $q^{-k}$.  Therefore
Claim~\ref{claim:merge-reduce-error} also gives the total-mass bound
\[
  q^{-k}W(H_B)
  \leq\mathbb E_{\boldsymbol\alpha}
       [\UnsatVal_{H_B}(\boldsymbol\alpha)]
  \leq b_\ell\mathbb E_{\boldsymbol\alpha}
       [\UnsatVal_{I_B}(\boldsymbol\alpha)]
  <2^{\ell+1}.
\]

Now, we consider a merge that creates a level-$\ell$ sparsifier, where $\ell\geq1$.  Its input $J$ is the union of two level-$(\ell-1)$ sparsifiers.  Adding the preceding bound for the two children gives
\begin{equation}
\label{eq:level-mass-bound}
  W(J)<2q^k2^\ell=O_R(2^\ell).
\end{equation}
Lemma~\ref{lem:integral-weight-reduction} and
\eqref{eq:error-schedule} now give
\[
\begin{split}
  |\operatorname{supp}(H_B)|
  &\leq
  c_{\mathrm{BG}}d\,
  \frac{\log_2^6(4\max\{W(J),1\})}{\eta_\ell^2}\\
  &\leq
  O_R\!\left(
    \frac{d}{\varepsilon^2}(\ell+1)^{10}
  \right),
\end{split}
\]
which is~\eqref{eq:level-support-bound}.

There is at most one sparsifier in each of $1+\lfloor\log_2\max\{m,1\}\rfloor$ levels.  Summing \eqref{eq:level-support-bound} shows that the number of stored constraints across all sparsifiers is
\[
  O_R\!\left(
    \frac{d}{\varepsilon^2}\log^{11}(m+2)
  \right).
\]
Each constraint stores an ordered $k$-tuple of variable indices and an integer weight. These use $O_R(\log(n+2)+\log(m+2))$ bits by \eqref{eq:level-weight-bound}.  The level counters and the error parameter use the remaining logarithmic term in \eqref{eq:stream-space-bound}.
\end{proof}

Finally, using the above, we can prove the theorem.

\begin{proof}[Proof of \Cref{thm:streaming-sparsification}]
If $m=0$, output the empty instance.  Otherwise, run the binary sparsify--merge--sparsify algorithm above and, after the last arrival, output the weighted union $H$ of the sparsifiers in the occupied levels.  Every retained constraint descends from a constraint in the stream, so $H\subseteq I$. Claim~\ref{claim:merge-reduce-error} gives, simultaneously for every assignment $\alpha$,
\[
  \UnsatVal_H(\alpha)
  \in(1\pm\varepsilon) \UnsatVal_I(\alpha).
\]
Claim~\ref{claim:merge-reduce-space} gives the explicit bound in~\eqref{eq:stream-space-bound}, and hence the space claimed in the theorem. Lemma~\ref{lem:integral-weight-reduction} makes every reduction deterministic. The input stream is read once. All later reductions operate only on stored sparsifiers. The algorithm need not know $m$ in advance because the error schedule depends only on the level.
\end{proof}

\begin{proof}[Proof of \Cref{cor:mincsp-main}]
Let $R^\star=\Core(R)$, with core domain $C$.  If $R^\star$ is trivial, then, since it is nonempty, $R^\star=C^k$.  Every $R^\star$-instance has minimum unsatisfied value zero, and the same is true for the corresponding $R$-instance by \Cref{lem:core-reduction}.  The upper bound is then immediate.  Hence suppose that $R^\star$ is non-trivial.

For the upper bound, replace every arriving $R$-constraint by the $R^\star$-constraint on the same scope, as in \Cref{lem:core-space}.  Set
\[
  \delta=\frac{\varepsilon}{2+\varepsilon}
\]
and run \Cref{thm:streaming-sparsification} with error $\delta$, obtaining a sparsifier $H$ of the resulting instance $I^\star$.  Minimizing the pointwise sparsifier inequalities gives
\begin{equation}
\label{eq:optimum-from-sparsifier}
  (1-\delta)\operatorname{OPT}_{\min}(I^\star)
  \leq\operatorname{OPT}_{\min}(H)
  \leq(1+\delta)\operatorname{OPT}_{\min}(I^\star).
\end{equation}
Output
\[
  Z=\frac{\operatorname{OPT}_{\min}(H)}{1-\delta}.
\]
\Cref{lem:core-reduction} gives $\operatorname{OPT}_{\min}(I^\star)=\operatorname{OPT}_{\min}(I)$, and therefore
\[
  \operatorname{OPT}_{\min}(I)
  \leq Z
  \leq
  \frac{1+\delta}{1-\delta}\operatorname{OPT}_{\min}(I)
  =(1+\varepsilon)\operatorname{OPT}_{\min}(I).
\]
The exact minimum of the stored sparsifier can be found after the stream by enumerating assignments; the model places no restriction on postprocessing time.  This uses $O_R(n)$ additional bits for one assignment, which is absorbed because $R^\star$ is non-trivial and $\NRD_n(R^\star)=\Omega_{R^\star}(n)$ by \cref{obs:nrd-linear}. \Cref{thm:streaming-sparsification} with parameter $\delta$ gives the space bound in part~(i).

For part~(ii), any finite-factor multiplicative approximation must output zero when $\operatorname{OPT}_{\min}=0$ and a positive number when $\operatorname{OPT}_{\min}>0$.  It therefore decides satisfiability.  The lower bound follows directly from \Cref{thm:sat-main}.
\end{proof}

\section{\texorpdfstring{\ExactCSP}{Exact-CSP}}
\label{sec:exact-csp}

We now determine the one-pass streaming space needed to solve $\ExactCSP(R)$ on simple instances. As in the previous sections, we first reduce $R$ to its core. After this reduction, the answer is governed by the degree of the core relation.

The two bounds use degree in different ways. For the upper bound, the degree-$d$ representation of $\mathbf 1_R$ lets us write the number of satisfied constraints using only terms that involve at most $d$ variables. The algorithm keeps the coefficient of each such term. When $d\ge1$, the lower bound uses $2^d$ assignments and a signed sum of their values that isolates one chosen constraint. The construction in \Cref{lem:exact-pinning} then lets Bob recover each of these values from the exact optimum returned by the streaming algorithm.

\begin{theorem}[$\ExactCSP$]
\label{thm:exact-main}
Let $R\subseteq D^k$, let $R^\star=\Core(R)$ have core domain $C$, and put $d=\Deg_C(\mathbf 1_{R^\star})$. The following statements hold.
\begin{enumerate}[label=(\roman*),leftmargin=*]
\item There is a deterministic one-pass algorithm for $\ExactCSP(R)$ using $O_R(n^d\log(n+1))$ bits on unweighted simple instances.
\item If $d\ge1$, every randomized bounded-error one-pass algorithm for $\ExactCSP(R)$ on unweighted simple instances uses $\Omega_R(n^d)$ bits of space.
\item If $R$ is symmetric and $d\ge1$, the lower bound holds even when no unordered scope is repeated.
\end{enumerate}
Consequently, when $d\ge1$,
\[
\Omega_R(n^d)
\le \SPACE(\ExactCSP(R))
\le O_R(n^d\log(n+1)).
\]
When $d=0$, the relation $R^\star$ is either empty or equal to $C^k$. Thus the optimum is either always $0$ or is the number of constraints in the stream, and it can be found using $O_R(\log(n+1))$ bits on simple instances.
\end{theorem}

We prove the upper bound directly. We then prove the lower bound for a core relation and use \Cref{lem:core-space} to obtain the result for an arbitrary relation.

The upper bound follows by keeping the coefficients of a low-degree representation of the objective. The Boolean version of this argument appears in \cite[Section~1.1]{KPSY23}. We give a generalization of their result to any finite domain.

\begin{lemma}[$\ExactCSP$ upper bound]
\label{lem:exact-upper}
Let $R\subseteq D^k$, let $R^\star=\Core(R)$ have core domain $C$, and put $d=\Deg_C(\mathbf 1_{R^\star})$. There is a deterministic one-pass algorithm that solves $\ExactCSP(R)$ on an instance with $n$ variables and $m$ constraints using $O_R(n^d(\log(n+1)+\log(m+1)))$ bits of space. In particular, on simple instances the space is $O_R(n^d\log(n+1))$.
\end{lemma}

\begin{proof}
By \Cref{lem:core-space}, we may replace each arriving $R$-constraint by the $R^\star$-constraint on the same scope. We will describe a small representation of the value of every assignment for this $R^\star$-instance.

Fix a basepoint $0\in C$. For a variable $x$ and a value $a\in C\setminus\{0\}$, define the indicator $X_{x,a}(\alpha)=\mathbf 1[\alpha(x)=a]$. This indicator is $1$ exactly when $\alpha$ gives $x$ the value $a$, and it is $0$ otherwise. By the definition of $d=\Deg_C(\mathbf 1_{R^\star})$, membership in $R^\star$ can be written as a polynomial in the indicators of coordinate values, with every monomial involving at most $d$ coordinates. In other words, there are constants $c_{T,b}$ such that, for every $z=(z_1,\ldots,z_k)\in C^k$,
\[
\mathbf 1_{R^\star}(z)=\sum_{\substack{T\subseteq[k]\\|T|\le d}}\ \sum_{b\in(C\setminus\{0\})^T}c_{T,b}\prod_{j\in T}\mathbf 1[z_j=b_j].
\]
The relation is fixed, so this sum has $O_R(1)$ terms and all the constants $c_{T,b}$ depend only on $R$. They are integers by the inclusion-exclusion description of the coefficients in \Cref{subsec:degree}.

Consider an arriving constraint $R^\star(x_{i_1},\ldots,x_{i_k})$. For this constraint, set $z_j=\alpha(x_{i_j})$ for every $j\in[k]$ in the formula above. We obtain
\[
\mathbf 1[R^\star(x_{i_1},\ldots,x_{i_k})\text{ is satisfied by }\alpha]=\sum_{\substack{T\subseteq[k]\\|T|\le d}}\ \sum_{b\in(C\setminus\{0\})^T}c_{T,b}\prod_{j\in T}X_{x_{i_j},b_j}(\alpha).
\]
This formula expresses whether $\alpha$ satisfies this constraint as a sum of products of the indicators $X_{x,a}$. Each product checks the values of at most $d$ variables from the scope. If the same variable appears twice in a product with the same required value, the two identical indicators can be replaced by one. If it appears with two different required values, no assignment can make both indicators equal to $1$, so the product is always zero and can be omitted.

The algorithm keeps one integer counter for every product of the form $\prod_{x\in U}X_{x,a_x}$ with $|U|\le d$. When a constraint arrives, the algorithm writes its satisfaction indicator using the formula above and adds each $c_{T,b}$ to the counter for the corresponding product. There are at most $\sum_{j=0}^d\binom{n}{j}(|C|-1)^j=O_R(n^d)$ possible products. Each counter has magnitude $O_R(m)$ and therefore uses $O_R(\log(m+1))$ bits. Identifying its variables and values uses $O_R(\log(n+1))$ bits. The total space is $O_R(n^d(\log(n+1)+\log(m+1)))$.

After the stream, the sum of the stored terms is exactly $\Val_{I^\star}(\alpha)$ for every assignment $\alpha$, where $I^\star$ is the resulting $R^\star$-instance. The algorithm can try all assignments, evaluate this sum, and return its maximum. There is no restriction on the time used after the stream. By \Cref{lem:core-reduction}, this maximum is also $\OPTmax(I)$ for the original $R$-instance $I$. Finally, a simple instance has at most $n^k$ constraints, so $\log(m+1)=O_R(\log(n+1))$.
\end{proof}

This proves the upper bound in \Cref{thm:exact-main}, including the case $d=0$. We now assume $d\ge1$ and prove the lower bound.

\subsection{Lower bound for a core relation}
\label{subsec:exact-lb}

We first prove the lower bound when $R$ is already a core. This is the main result of the subsection.

\begin{theorem}[Lower bound for $\ExactCSP$ on a core relation]
\label{thm:exact-core-lb}
Let $R\subseteq D^k$ be a core relation and let $d=\Deg_D(\mathbf 1_R)\ge 1$. Every randomized bounded-error one-pass algorithm for $\ExactCSP(R)$ on unweighted simple instances uses $\Omega_R(n^d)$ bits of space. If $R$ is symmetric, the statement holds with unordered simplicity.
\end{theorem}

The proof is a reduction from $\Index_{N^d}$. Alice will associate one constraint with each vector in $[N]^d$ and stream that constraint exactly when the corresponding bit is $1$. Bob must determine whether Alice streamed one chosen constraint. To do this, he will consider $2^d$ assignments and add or subtract their values. The next lemma shows that these signs can be chosen so that the selected constraint has a nonzero total contribution. Later, we will arrange that every other constraint has total contribution zero.

More precisely, each of the $d$ relevant coordinates will take either a basepoint $0$ or one fixed value $a_i$. The assignments are indexed by the sets $S\subseteq[d]$ of coordinates that receive $a_i$. The expression in the next lemma is simply the value on these $2^d$ choices, added with a plus sign when $d-|S|$ is even and a minus sign when it is odd.

\begin{lemma}
\label{lem:finite-difference}
Let $R\subseteq D^k$ and suppose $\Deg_D(\mathbf 1_R)=d\ge 1$. Fix any basepoint $0\in D$. Then there are distinct coordinates $t_1,\ldots,t_d\in[k]$ and values $a_1,\ldots,a_d\in D\setminus\{0\}$ such that, after fixing every coordinate outside $\{t_1,\ldots,t_d\}$ to $0$, the resulting function $h:D^d\to\{0,1\}$ satisfies
\[
c:=\sum_{S\subseteq[d]}(-1)^{d-|S|}h(y^S)\neq 0,
\]
where $y_i^S=a_i$ for $i\in S$ and $y_i^S=0$ otherwise.
\end{lemma}

\begin{proof}
For $j\in[k]$ and $a\in D\setminus\{0\}$, define $Z_{j,a}(x)=\mathbf 1[x_j=a]$. As in the upper bound, the definition of degree gives an expansion
\[
\mathbf 1_R(x)=\sum_{T\subseteq[k]}\ \sum_{b\in(D\setminus\{0\})^T}c_{T,b}\prod_{j\in T}Z_{j,b_j}(x).
\]
Since the degree is $d$, there is a set $T=\{t_1,\ldots,t_d\}$ and a choice of values $a_1,\ldots,a_d\in D\setminus\{0\}$ for which the coefficient $c_{T,a}$ is nonzero. Fix every coordinate outside $T$ to $0$. Any term involving one of those coordinates is now zero because all the indicators in the expansion correspond to nonzero values.

Consider the signed sum in the statement one term of the expansion at a time. If a term does not use coordinate $t_i$, then its value is unchanged when $i$ is added to or removed from $S$. The sets $S$ and $S\mathbin\triangle\{i\}$ have opposite signs, so their contributions cancel. Therefore, only terms that use all $d$ coordinates in $T$ can remain. There are no terms of degree larger than $d$, and among the degree-$d$ terms on $T$, only the term requiring the values $a_1,\ldots,a_d$ is nonzero on the chosen $2^d$ tuples. Its total contribution is its coefficient $c_{T,a}$. Hence the number $c$ in the statement equals $c_{T,a}$ and is nonzero.
\end{proof}

We now use these $2^d$ choices in the reduction. The queried constraint will see all of them and will therefore contribute $c$. Any other constraint will be unaffected by at least one of the choices. Pairing the two assignments that differ only in that choice will make its contribution cancel.

\begin{proof}[Proof of \Cref{thm:exact-core-lb}.]
Without loss of generality, we may assume that the coordinates $t_1,\ldots,t_d$ from \Cref{lem:finite-difference} are the first $d$ coordinates. This is because if they are not, we can permute the coordinates of $R$ and apply the same permutation to every scope without changing the streaming problem. Fix the values $a_1,\ldots,a_d$ and the nonzero integer $c$ given by that lemma. If $k=1$ and $|D|>1$, a unary relation always has a non-bijective endomorphism. For a nonempty relation, use a constant map to a value in the relation. For the empty relation, use any non-bijective map. Thus a unary core relation has a one-element domain and degree zero. It follows that $d\ge1$ implies $k\ge2$, as required in \Cref{lem:exact-pinning}.

Let $N$ be a positive integer and index the $N^d$ bits of $\Index_{N^d}$ by the vectors $u=(u_1,\ldots,u_d)\in[N]^d$. For each $i\in[d]$, create a group of variables $X_i=\{x_{i,p}:p\in[N]\}$. The variable $x_{i,u_i}$ will occupy coordinate $i$ of the constraint associated with $u$.

We also need variables for the coordinates in $[k]\setminus[d]$. These coordinates must receive the basepoint $0$ in Bob's assignments. Using one variable in every constraint would give that variable degree $N^d$, which is too large for \Cref{lem:exact-pinning}. We instead create variables $w_{j,p}$ for every $j\in[k]\setminus[d]$ and $p\in[N]$, and use $w_{j,u_1}$ in coordinate $j$ of the constraint associated with $u$.

Formally, for $u\in[N]^d$, define $C_u=R(v_1^u,\ldots,v_k^u)$ by setting $v_i^u=x_{i,u_i}$ for $i\in[d]$ and $v_j^u=w_{j,u_1}$ for $j\in[k]\setminus[d]$.

Alice receives $b\in\{0,1\}^{[N]^d}$ and streams $C_u$ exactly when $b_u=1$. Let $P_b$ be the resulting instance. It uses at most $kN=O_R(N)$ variables and at most $N^d$ constraints. Every variable belongs to at most $N^{d-1}$ of the constraints $C_u$. Fixing the index of that variable fixes at least one coordinate of $u$. The constraints have distinct ordered scopes because the variables from the groups $X_1,\ldots,X_d$ determine $u$. They also have distinct unordered scopes. Indeed, the groups $X_1,\ldots,X_d$ are disjoint. If $C_u$ and $C_{u'}$ had the same unordered scope, then for every $i\in[d]$ they would contain the same variable from $X_i$. This would give $x_{i,u_i}=x_{i,u_i'}$, and hence $u_i=u_i'$, for every $i$. Therefore $u=u'$. This last fact is relevant when $R$ is symmetric.

Bob receives $q=(q_1,\ldots,q_d)\in[N]^d$ and wants to recover $b_q$. His goal is to make $C_q$ see all $2^d$ tuples from \Cref{lem:finite-difference}. At the same time, every constraint $C_u$ with $u\ne q$ should miss at least one of the $d$ choices, so that its positive and negative contributions cancel. To achieve this, Bob changes only the variable $x_{i,q_i}$ in each group $X_i$ and leaves every other variable at the basepoint.

For each $S\subseteq[d]$, Bob considers the following assignment. In group $X_i$, the variable indexed by $q_i$ receives $a_i$ if $i\in S$, and all other variables receive $0$. Every variable $w_{j,p}$ also receives $0$. In symbols,
\[
\alpha^S(x_{i,p})=
\begin{cases}
a_i,&\text{if }i\in S\text{ and }p=q_i,\\
0,&\text{otherwise,}
\end{cases}
\qquad
\alpha^S(w_{j,p})=0.
\]
Consider first the queried constraint $C_q$. In coordinate $i\in[d]$, it receives $a_i$ when $i\in S$ and $0$ otherwise, while every coordinate in $[k]\setminus[d]$ receives $0$. Thus, as $S$ ranges over all subsets of $[d]$, this constraint sees exactly the $2^d$ tuples from \Cref{lem:finite-difference}, and its signed contribution is $c$.

Now consider $C_u$ for $u\ne q$. Choose $i$ for which $u_i\ne q_i$. The variable $x_{i,u_i}$ always receives $0$, whether or not $i$ belongs to $S$. Hence the satisfaction of $C_u$ is unchanged when we replace $S$ by $S\mathbin\triangle\{i\}$. These two sets have opposite signs, so the contribution of $C_u$ cancels. We have proved that
\[
\sum_{S\subseteq[d]}(-1)^{d-|S|}\mathbf 1[C_u\text{ is satisfied by }\alpha^S]
=
\begin{cases}
c,&\text{if }u=q,\\
0,&\text{if }u\ne q.
\end{cases}
\]
Summing this identity over the constraints that Alice streamed gives
\[
\sum_{S\subseteq[d]}(-1)^{d-|S|}\Val_{P_b}(\alpha^S)=c\,b_q.
\]

This identity determines $b_q$ from the $2^d$ numbers $\Val_{P_b}(\alpha^S)$ because $c\ne0$. The difficulty is that an $\ExactCSP$ algorithm returns an optimum, not the value of a chosen assignment. By \Cref{lem:exact-pinning}, Bob can append constraints that hard-pin a complete assignment and make the optimum reveal any one of these values. More precisely, for each $S\subseteq[d]$, \Cref{lem:exact-pinning} gives an instance $H_{\alpha^S}$ satisfying
\[
\OPTmax(P_b\cup H_{\alpha^S})=|H_{\alpha^S}|+\Val_{P_b}(\alpha^S).
\]
By the final part of \Cref{lem:exact-pinning}, Bob can construct $H_{\alpha^S}$ from the public variable set, the assignment $\alpha^S$, and the bounds $N^d$ and $N^{d-1}$. He does not need to know Alice's string. The gadget uses $O_R(N)$ new variables, and $P_b\cup H_{\alpha^S}$ is simple. When $R$ is symmetric, it has no repeated unordered scope.

Suppose now that the streaming algorithm uses $s(n)$ bits of space. We first reduce its error so that each output is wrong with probability at most $1/(10\cdot2^d)$. This requires only $O_R(1)$ independent copies because $d$ depends only on $R$. Alice runs the amplified algorithm on $P_b$ and sends its memory state to Bob. For each $S\subseteq[d]$, Bob starts from a copy of this state, appends $H_{\alpha^S}$, and obtains $\OPTmax(P_b\cup H_{\alpha^S})$. Subtracting the known number $|H_{\alpha^S}|$ gives $\Val_{P_b}(\alpha^S)$. A union bound shows that all $2^d$ values are correct simultaneously with constant probability, after which Bob recovers $b_q$ from the signed sum above.

There is a constant $A_R$ such that every final instance uses at most $A_RN$ variables. For any sufficiently large target value $n$, choose $N=\lfloor n/A_R\rfloor$ and add unused variables so that the final instance has exactly $n$ variables. Then $N=\Theta_R(n)$. We have obtained a randomized public-coin one-way protocol for $\Index_{N^d}$ with $O_R(s(n))$ communication. By \Cref{lem:index}, $s(n)=\Omega_R(N^d)=\Omega_R(n^d)$.
\end{proof}

We have proved the required lower bound when $R$ is a core. It remains only to combine this with the reduction to the core and the upper bound proved above.

\begin{proof}[Proof of \Cref{thm:exact-main}.]
Let $R^\star=\Core(R)$. By \Cref{lem:core-space}, solving $\ExactCSP(R)$ has the same streaming space complexity as solving $\ExactCSP(R^\star)$. The deterministic upper bound follows from \Cref{lem:exact-upper}. If $d=0$, the indicator of $R^\star$ is constant, so $R^\star$ is either empty or equal to $C^k$. This gives the last statement of the theorem. Suppose now that $d\ge1$. The randomized lower bound follows by applying \Cref{thm:exact-core-lb} to $R^\star$. The reduction to the core preserves simple instances. If $R$ is symmetric, then $R^\star$ is symmetric and the reduction also preserves the unordered scopes. This proves all three parts of the theorem.
\end{proof}

Thus the streaming complexity of $\ExactCSP(R)$ on simple instances is governed, up to the logarithmic factor in the upper bound, by the degree of the core relation.

\section{Conclusion and Open Questions}
\label{sec:conclusion}

In this paper, we demonstrated that the core of a relation is vital for characterizing the streaming complexity of many optimization problems related to constraint satisfaction problems. We hope that this work encourages further exploration of the relationship between \CSP cores and streaming complexity. We conclude with a handful of open questions.

\begin{itemize}
\item Can the algorithms presented in \Cref{thm:sat-main} and \Cref{thm:streaming-sparsification} be made efficient in terms of both space and time complexity? A key challenge in this regard is that there is no general technique for efficiently detecting if a \CSP instance is non-redundant. If one could find such an algorithm, it would be a major milestone in the development of algorithms for \CSP kernelization. See \cite{chen2020best,lagerkvist2020sparsification,carbonnel2022redundancy,brakensiek2025redundancy} for further discussion.

\item In addition, our paper only considers \emph{insertion-only} streams. What if we allow both insertions and deletions? 
\end{itemize}

\printbibliography

\appendix

\section{Literals and Non-Constant-Satisfiable \texorpdfstring{\CSPs}{CSPs}}
\label{sec:literals-and-constants}

In this section, we show how our general \CSPs core characterization recovers the two earlier results of Sharma and Velusamy~\cite{sharma2026characterizingstreamingdecidabilitycsps}. We first consider \CSPs with literals and show that passing their shift/sign-closed language to the core does not change non-redundancy. We then show that a positive Boolean constraint language is a core exactly when it is not constant-satisfiable.

A \emph{constraint language} $\Gamma$ is a fixed finite set of relations over the same domain. A map on the domain is an endomorphism of $\Gamma$ if it preserves every relation in $\Gamma$, and $\Gamma$ is a core if every such endomorphism is an automorphism. As for one relation, we obtain $\Core(\Gamma)$ by choosing an endomorphism $\rho$ with minimum image and restricting every relation to that image.

Without loss of generality, we claim that the core reduction and the satisfiability arguments also work for a fixed finite constraint language. In the core reduction, each arriving $S$-constraint, for $S\in\Gamma$, is replaced by the $\rho(S)$-constraint on the same scope. The satisfiability hard-pinning construction may similarly use constraints from any relation in $\Gamma$. We will use this extension below without further comment.

\subsection{Literal \texorpdfstring{\CSPs}{CSPs} as a special case}
\label{sec:literal-core}

We first show that our general \CSPs core characterization recovers the previous characterization for \CSPs with literals. In the literal model, the domain is $D=\mathbb{Z}_q$, and every constraint may apply an arbitrary additive shift to each coordinate. For any relation $S\subseteq D^\ell$ and any $\lambda\in D^\ell$, write $S^\lambda(x):=S(x+\lambda)$. Thus, for a relation $R\subseteq D^k$, its \emph{shift closure} is
\[
\Gamma_R^{\mathrm{lit}}:=\{R^\lambda:\lambda\in D^k\}.
\]
Then a \CSP with literals over $R$ is exactly a positive \CSP over $\Gamma_R^{\mathrm{lit}}$. We write $\NRD_n^+(\Gamma)$ for the largest number of constraints in a nonredundant positive $\Gamma$-instance on $n$ variables. Similarly, $\NRD_n^{\mathrm{lit}}(R)$ denotes non-redundancy in the literal model over $R$.

For a constraint language $\Gamma$, let $\SAT(\Gamma)$ denote satisfiability for positive $\Gamma$-instances. The following corollary records the language version of our satisfiability characterization.

\begin{corollary}[Satisfiability for constraint languages]
\label{cor:sat-language}
For every fixed finite constraint language $\Gamma$,
\[
\Omega_\Gamma\!\left(\NRD_n^+(\Core(\Gamma))\right)
\le \SPACE(\SAT(\Gamma))
\le O_\Gamma\!\left((\NRD_n^+(\Core(\Gamma))+1)\log(n+1)\right).
\]
The upper bound is deterministic, and the lower bound holds for randomized bounded-error one-pass algorithms.
\end{corollary}

\begin{proof}
The core reduction replaces each relation in $\Gamma$ by its restriction to the core domain and preserves satisfiability and streaming space exactly, as in \Cref{lem:core-space}. For a core language, the upper bound maintains a nonredundant subinstance with the same satisfying assignments, exactly as in \Cref{lem:sat-upper}. If the language contains a nonempty proper relation, disjoint constraints from this relation show that its non-redundancy is $\Omega_\Gamma(n)$, so the space needed to try assignments is within the stated bound. Otherwise every relation is empty or full, and satisfiability can be decided by checking whether an empty-relation constraint appears. For the lower bound, the hard-pinning construction and the $\Index$ reduction in \Cref{thm:sat-core-lb} work with constraints from any relation in the core language. This proves the two bounds.
\end{proof}

The following lemma shows that passing to the core does not change non-redundancy for any shift-closed language.

\begin{lemma}
\label{lem:shift-core-nrd}
Let $\Gamma$ be a finite constraint language over $D=\mathbb{Z}_q$ that is closed under coordinatewise shifts. Then
\(
\NRD_n^+(\Gamma)=\NRD_n^+(\Core(\Gamma)).
\)
\end{lemma}

\begin{proof}
First, we may assume that every relation in $\Gamma$ is nonempty. Empty relations are fixed by every shift and every endomorphism. Moreover, if a nonredundant instance contains an empty-relation constraint, then it can contain no other constraint. Empty relations therefore give the same trivial size-one instances before and after passing to the core, so we may discard them. If no relations remain, the statement is immediate. 

For a relation $S$ and a coordinate $i$, let $S+d e_i$ denote the relation obtained by adding $d$ to coordinate $i$, where $e_i$ is the $i$th unit vector. Define
\[
H:=\{d\in D:S+d e_i=S\text{ for every }S\in\Gamma\text{ and every coordinate }i\}.
\]
Thus $H$ is the subgroup of shifts that are invisible to every relation.

Choose one representative from each coset of $H$, let $C\subseteq D$ be the set of representatives, and let $\rho:D\to C$ map every element to the representative of its coset. Since $\rho(a)-a\in H$, for every $S\in\Gamma$ and every tuple $x$,
\begin{equation}
x\in S\quad\Longleftrightarrow\quad\rho(x)\in S.
\label{eq:shift-membership}
\end{equation}
In particular, $\rho$ is an endomorphism.

We claim that $\rho$ is a core endomorphism. Let $h:D\to D$ be any endomorphism and write $\delta(a):=h(a)-a$. We show that
\begin{equation}
\delta(a)-\delta(b)\in H\qquad\text{for all }a,b\in D.
\label{eq:shift-delta}
\end{equation}
Indeed, fix $S\in\Gamma$, $x\in S$, and an arbitrary tuple $z$. Write $\delta(z)$ for the tuple obtained by applying $\delta$ to every coordinate of $z$. Since $\Gamma$ is shift-closed, $S^{x-z}\in\Gamma$, and $z\in S^{x-z}$. Applying $h$ gives $h(z)\in S^{x-z}$, or equivalently $x+\delta(z)\in S$. Thus translation by $\delta(z)$ maps $S$ into itself. Since translation is a bijection and $S$ is finite, it preserves $S$. Comparing two tuples $z,z'$ that differ only in coordinate $i$, taking the values $a$ and $b$ there, shows that translation by $(\delta(a)-\delta(b))e_i$ preserves $S$. Since $S$ and $i$ were arbitrary, \eqref{eq:shift-delta} follows.

Now suppose that $h(a)=h(b)$. Then $b-a=\delta(a)-\delta(b)\in H$. Hence an endomorphism cannot identify elements from distinct $H$-cosets. Every endomorphism therefore has image size at least $|D/H|$, while $\rho$ has exactly this image size. Thus $\rho$ is a core endomorphism.

Let $\Delta:=\rho(\Gamma)=\Core(\Gamma)$. For $S\in\Gamma$, write $S^\rho:=\rho(S)$. By \eqref{eq:shift-membership},
\begin{equation}
x\in S\quad\Longleftrightarrow\quad\rho(x)\in S^\rho.
\label{eq:shift-constraintwise}
\end{equation}
Moreover, the map $S\mapsto S^\rho$ is injective. If $S^\rho=T^\rho$, then \eqref{eq:shift-constraintwise} implies $S=T$.

Now take any positive $\Gamma$-instance $I$, and replace every relation $S$ by $S^\rho$, obtaining a $\Delta$-instance $I^\rho$. This correspondence preserves the number of constraints. Moreover, \eqref{eq:shift-constraintwise} shows that an assignment and its image under $\rho$ have the same satisfaction pattern, constraint by constraint. Consequently, if $a$ is a witness for a constraint in $I$, then $\rho(a)$ is a witness for the corresponding constraint in $I^\rho$. Hence $I^\rho$ is nonredundant whenever $I$ is nonredundant.

Conversely, every $\Delta$-instance can be lifted to a $\Gamma$-instance using the unique preimage of each relation, and every witness assignment over $C$ remains a witness after this lift. Thus nonredundant instances of the same size exist in both directions, proving $\NRD_n^+(\Gamma)=\NRD_n^+(\Delta)$.
\end{proof}

Applying the lemma to the shift closure gives the desired connection to the literal model.

\begin{corollary}
\label{cor:literal-core-nrd}
For every $R\subseteq\mathbb{Z}_q^k$,
\(
\NRD_n^{\mathrm{lit}}(R)=\NRD_n^+\!\left(\Core(\Gamma_R^{\mathrm{lit}})\right).
\)
Hence our general core characterization recovers the previous streaming characterization for \CSPs with literals.
\end{corollary}

\begin{proof}
A literal $R$-instance is exactly a positive $\Gamma_R^{\mathrm{lit}}$-instance, so $\NRD_n^{\mathrm{lit}}(R)=\NRD_n^+(\Gamma_R^{\mathrm{lit}})$. Since $\Gamma_R^{\mathrm{lit}}$ is shift-closed, \Cref{lem:shift-core-nrd} gives
\(
\NRD_n^+(\Gamma_R^{\mathrm{lit}})=\NRD_n^+\!\left(\Core(\Gamma_R^{\mathrm{lit}})\right).
\)
Together with \Cref{cor:sat-language}, this recovers the one-pass satisfiability characterization for \CSPs with literals.
\end{proof}

\subsection{Non-constant-satisfiable Boolean \texorpdfstring{\CSPs}{CSPs}}
\label{sec:boolean-nonconstant}

We next recover the previous characterization for positive Boolean \CSPs. A Boolean constraint language $\Gamma$ is \emph{constant-satisfiable} if there is a value $c\in\{0,1\}$ such that the constant tuple $(c,\ldots,c)$ belongs to every relation in $\Gamma$. Equivalently, assigning the value $c$ to every variable satisfies every positive instance over $\Gamma$. The next lemma shows that excluding this case is exactly the same as requiring the language to be a core.

\begin{lemma}
\label{lem:boolean-languages-are-cores}
Let $\Gamma$ be a Boolean constraint language in which every relation is nonempty.
Then $\Gamma$ is a core if and only if it is not constant-satisfiable.
\end{lemma}

\begin{proof}
Every non-bijective map from $\{0,1\}$ to itself is constant. The constant map with value $c$ is an endomorphism of $\Gamma$ exactly when $(c,\ldots,c)$ belongs to every relation in $\Gamma$. One direction is immediate. For the other, apply the constant map to any tuple in each nonempty relation. Thus $\Gamma$ has a non-bijective endomorphism exactly when it is constant-satisfiable.
\end{proof}

It follows that every non-constant-satisfiable Boolean constraint language is already a core. Applying our general \CSP core characterization therefore directly recovers the positive Boolean result of Sharma and Velusamy.

\end{document}

%% file: header.tex
\usepackage[margin=1in]{geometry}

\usepackage[utf8]{inputenc}

\usepackage{amsthm,amsmath,amssymb,mathtools}
\usepackage{bm}
\usepackage{graphicx}
\usepackage{color,xcolor}
\usepackage{paralist,enumitem}
\usepackage[normalem]{ulem}
\usepackage{tabularx}
\usepackage{tikz}
\usepackage{caption,comment}
\usepackage{algorithmicx}
\usepackage{algorithm}
\usepackage{algpseudocode}
\usepackage{lipsum}
\usepackage{booktabs}
\usepackage{subcaption}
\usepackage{xspace}
\usepackage{lineno}

\nolinenumbers
\usepackage{microtype}
\usepackage{silence}
\newcommand{\defref}[2]{%
  {\hypersetup{linkcolor=green!50!black}%
  \hyperref[#2]{\emph{#1}}}%
}

\newcounter{algsubstate}

\algnewcommand\algorithmicinput{\textbf{Input:}}
\algnewcommand\Input{\item[\algorithmicinput]}
\algnewcommand\algorithmicoutput{\textbf{Output:}}
\algnewcommand\Output{\item[\algorithmicoutput]}
\algnewcommand\algorithmicgoal{\textbf{Goal:}}
\algnewcommand\Goal{\item[\algorithmicgoal]}

\usepackage[colorinlistoftodos,textsize=tiny,textwidth=2cm,color=green!50!gray]{todonotes}
\newif\iftodoss
\todossfalse % final version

\newcommand{\DeclareAuthorCommands}[3]{\iftodoss\expandafter\newcommand\csname #1\endcsname[1]{\textcolor{#3}{##1}}\expandafter\newcommand\csname #1i\endcsname[1]{\textcolor{#3}{[#2: ##1]}}\expandafter\newcommand\csname #1r\endcsname[1]{\todo[color=#3!20]{#2: ##1}}\else\expandafter\newcommand\csname #1\endcsname[1]{##1}\expandafter\newcommand\csname #1i\endcsname[1]{}\expandafter\newcommand\csname #1r\endcsname[1]{}\fi}
\DeclareAuthorCommands{amat}{AS}{brown}
\DeclareAuthorCommands{sant}{SV}{red}
\DeclareAuthorCommands{loui}{AP}{blue}
\DeclareAuthorCommands{jos}{JB}{green!50!black}

\usepackage[style=alphabetic,backend=biber,minalphanames=3,maxalphanames=4,maxbibnames=99,maxcitenames=4]{biblatex}

\usepackage{thmtools,thm-restate}
\usepackage[colorlinks=true,allcolors=blue]{hyperref}

\usepackage[capitalize,noabbrev,nameinlink]{cleveref}
\numberwithin{equation}{section}

\declaretheoremstyle[
    spaceabove=12pt,
    spacebelow=12pt,
    bodyfont=\normalfont,
    qed=\ensuremath{\diamond}
]{examplestyle}

\declaretheoremstyle[bodyfont=\it,qed=\qedsymbol]{noproofstyle}

\declaretheorem[numberlike=equation]{observation}

\declaretheorem[name=Observation,numbered=no]{observation*}
\declaretheorem[numberlike=equation]{fact}

\declaretheorem[numberlike=equation]{theorem}

\declaretheorem[name=Theorem,numbered=no]{theorem*}
\declaretheorem[numberlike=equation]{lemma}
\declaretheorem[name=Lemma,numbered=no]{lemma*}

\declaretheorem[numberlike=equation]{corollary}
\declaretheorem[name=Corollary,numbered=no]{corollary*}

\declaretheorem[name=Proposition,numbered=no]{proposition*}

\declaretheorem[numberlike=equation]{claim}
\declaretheorem[name=Claim,numbered=no]{claim*}

\declaretheorem[name=Conjecture,numbered=no]{conjecture*}

\declaretheorem[name=Question,numbered=no]{question*}

\declaretheoremstyle[spaceabove=12pt,spacebelow=12pt,bodyfont=\it]{defstyle}

\declaretheorem[unnumbered,name=Definition,style=defstyle]{definition*}
\declaretheoremstyle[spaceabove=12pt,spacebelow=12pt,bodyfont=\normalfont]{examplestyle}
\declaretheorem[numberlike=equation,style=examplestyle]{example}
\declaretheorem[unnumbered,name=Example,style=examplestyle]{example*}

\declaretheorem[unnumbered,name=Notation,style=defstyle]{notation*}

\declaretheorem[unnumbered,name=Construction,style=defstyle]{construction*}

\crefname{claim}{claim}{claims}
\crefname{fact}{fact}{facts}
\crefname{corollary}{corollary}{corollaries}
\crefname{theorem}{theorem}{theorems}
\crefname{table}{table}{tables}
\Crefname{theorem}{Theorem}{Theorems}
\Crefname{claim}{Claim}{Claims}
\Crefname{fact}{Fact}{Facts}
\Crefname{corollary}{Corollary}{Corollaries}
\Crefname{table}{Table}{Tables}

\newcommand{\poly}{\operatorname{poly}}

\allowdisplaybreaks

%% file: refs.bib
@inproceedings{brakensiek2026contraction,
  author    = {Brakensiek, Joshua and Guruswami, Venkatesan and Putterman, Aaron},
  title     = {Multiplicative Error Set System Sparsification:
               A Simpler Proof via Chain Length Contraction},
  booktitle = {53rd International Colloquium on Automata, Languages,
               and Programming ({ICALP} 2026)},
  series    = {LIPIcs},
  volume    = {374},
  pages     = {44:1--44:17},
  publisher = {Schloss Dagstuhl -- Leibniz-Zentrum f{\"u}r Informatik},
  year      = {2026},
  doi       = {10.4230/LIPIcs.ICALP.2026.44}
}

@article{brakensiek2026separating,
  author  = {Brakensiek, Joshua and Guruswami, Venkatesan and Putterman, Aaron},
  title   = {Separating Non-redundancy and Chain Length},
  journal = {arXiv preprint arXiv:2609.17914},
  year    = {2026},
  url     = {https://arxiv.org/abs/2609.17914}
}

@article{brakensiek2026abelian,
  author  = {Brakensiek, Joshua and Guruswami, Venkatesan and Putterman, Aaron},
  title   = {The Reach of {Abelian} Covers in Hypergraphs},
  journal = {arXiv preprint arXiv:2609.36761},
  year    = {2026},
  url     = {https://arxiv.org/abs/2609.36761}
}

@article{mcgregor2014graph,
  title={Graph stream algorithms: a survey},
  author={McGregor, Andrew},
  journal={ACM SIGMOD Record},
  volume={43},
  number={1},
  pages={9--20},
  year={2014},
  publisher={ACM New York, NY, USA}
}

@InProceedings{sharma2026characterizingstreamingdecidabilitycsps,
  author =	{Sharma, Amatya and Velusamy, Santhoshini},
  title =	{{Characterizing Streaming Decidability of CSPs via Non-Redundancy}},
  booktitle =	{34th Annual European Symposium on Algorithms (ESA 2026)},
  pages =	{131:1--131:20},
  series =	{Leibniz International Proceedings in Informatics (LIPIcs)},
  ISBN =	{978-3-95977-445-1},
  ISSN =	{1868-8969},
  year =	{2026},
  volume =	{388},
  editor =	{Bille, Philip and Pettie, Seth and Storandt, Sabine},
  publisher =	{Schloss Dagstuhl -- Leibniz-Zentrum f{\"u}r Informatik},
  address =	{Dagstuhl, Germany},
  URL =		{https://drops.dagstuhl.de/entities/document/10.4230/LIPIcs.ESA.2026.131},
  URN =		{urn:nbn:de:0030-drops-272679},
  doi =		{10.4230/LIPIcs.ESA.2026.131}
}

@article{filtser2017sparsification,
  title={Sparsification of two-variable valued constraint satisfaction problems},
  author={Filtser, Arnold and Krauthgamer, Robert},
  journal={SIAM Journal on Discrete Mathematics},
  volume={31},
  number={2},
  pages={1263--1276},
  year={2017},
  publisher={SIAM}
}

@inproceedings{khanna2025efficient,
  title={Efficient algorithms and new characterizations for {CSP} sparsification},
  author={Khanna, Sanjeev and Putterman, Aaron and Sudan, Madhu},
  booktitle={Proceedings of the 57th Annual ACM Symposium on Theory of Computing},
  pages={407--416},
  year={2025}
}

@inproceedings{brakensiek2025redundancy,
  title={Redundancy is all you need},
  author={Brakensiek, Joshua and Guruswami, Venkatesan},
  booktitle={Proceedings of the 57th Annual ACM Symposium on Theory of Computing},
  pages={1614--1625},
  year={2025}
}

@article{LovettMW26,
  title={Moonflowers and Efficient Code Sparsification},
  author={Lovett, Shachar and Meka, Raghu and Wang, Yimeng},
  journal={arXiv preprint arXiv:2605.08676},
  year={2026}
}

@article{KhannaPS26,
  title={A Unified Theory of Sparsification},
  author={Khanna, Sanjeev and Putterman, Aaron and Sudan, Madhu},
  journal={arXiv preprint arXiv:2607.16126},
  year={2026}
}

@inproceedings{kogan2015sketching,
  title={Sketching cuts in graphs and hypergraphs},
  author={Kogan, Dmitry and Krauthgamer, Robert},
  booktitle={Proceedings of the 2015 Conference on Innovations in Theoretical Computer Science},
  pages={367--376},
  year={2015}
}

@article{jansen2019optimal,
  title={Optimal sparsification for some binary CSPs using low-degree polynomials},
  author={Jansen, Bart MP and Pieterse, Astrid},
  journal={ACM Transactions on Computation Theory (TOCT)},
  volume={11},
  number={4},
  pages={1--26},
  year={2019},
  publisher={ACM New York, NY, USA}
}

@article{lagerkvist2020sparsification,
  title={Sparsification of SAT and CSP problems via tractable extensions},
  author={Lagerkvist, Victor and Wahlstr{\"o}m, Magnus},
  journal={ACM Transactions on Computation Theory (TOCT)},
  volume={12},
  number={2},
  pages={1--29},
  year={2020},
  publisher={ACM New York, NY, USA}
}

@article{brakensiek2025richness,
  title={The Richness of CSP Non-redundancy},
  author={Brakensiek, Joshua and Guruswami, Venkatesan and Jansen, Bart MP and Lagerkvist, Victor and Wahlstr{\"o}m, Magnus},
  journal={arXiv preprint arXiv:2507.07942},
  year={2025}
}

@article{chen2020best,
  title={Best-case and worst-case sparsifiability of Boolean {CSPs}},
  author={Chen, Hubie and Jansen, Bart MP and Pieterse, Astrid},
  journal={Algorithmica},
  volume={82},
  number={8},
  pages={2200--2242},
  year={2020},
  publisher={Springer}
}

@article{Index2,
  author    = {Ilan Kremer and Noam Nisan and Dana Ron},
  title     = {On randomized one-round communication complexity},
  journal   = {Computational Complexity},
  volume    = {8},
  number    = {1},
  pages     = {21--49},
  year      = {1999},
  doi       = {10.1007/PL00001596}
}

@article{Vu-2208.09160,
  title={Revisiting maximum satisfiability and related problems in data streams},
  author={Vu, Hoa T},
  journal={Theoretical Computer Science},
  volume={982},
  pages={114271},
  year={2024},
  publisher={Elsevier}
}

@inproceedings{KPS24,
  author       = {Sanjeev Khanna and
                  Aaron (Louie) Putterman and
                  Madhu Sudan},
  editor       = {David P. Woodruff},
  title        = {Code Sparsification and its Applications},
  booktitle    = {Proceedings of the 2024 {ACM-SIAM} Symposium on Discrete Algorithms,
                  {SODA} 2024, Alexandria, VA, USA, January 7-10, 2024},
  pages        = {5145--5168},
  publisher    = {{SIAM}},
  year         = {2024},
  url          = {https://doi.org/10.1137/1.9781611977912.185},
  doi          = {10.1137/1.9781611977912.185},
  bibsource    = {dblp computer science bibliography, https://dblp.org}
}

@inproceedings{carbonnel2022redundancy,
  title={On redundancy in constraint satisfaction problems},
  author={Carbonnel, Cl{\'e}ment},
  booktitle={28th International Conference on Principles and Practice of Constraint Programming (CP 2022)},
  pages={11:1--11:15},
  year={2022},
  organization={Schloss Dagstuhl--Leibniz-Zentrum f{\"u}r Informatik}
}

@inproceedings{bessiere2020chain,
  title={Chain length and CSPs learnable with few queries},
  author={Bessiere, Christian and Carbonnel, Cl{\'e}ment and Katsirelos, George},
  booktitle={Proceedings of the AAAI Conference on Artificial Intelligence},
  volume={34-02},
  pages={1420--1427},
  year={2020}
}

@article{feder1998computational,
  title={The computational structure of monotone monadic SNP and constraint satisfaction: A study through Datalog and group theory},
  author={Feder, Tom{\'a}s and Vardi, Moshe Y},
  journal={SIAM Journal on Computing},
  volume={28},
  number={1},
  pages={57--104},
  year={1998},
  publisher={SIAM}
}

@article{zhuk2020proof,
  title={A proof of the CSP dichotomy conjecture},
  author={Zhuk, Dmitriy},
  journal={Journal of the ACM (JACM)},
  volume={67},
  number={5},
  pages={1--78},
  year={2020},
  publisher={ACM New York, NY, USA}
}

@inproceedings{GVV17,
  title = {Streaming {{Complexity}} of {{Approximating Max 2CSP}} and {{Max Acyclic Subgraph}}},
  booktitle = {Approximation, Randomization, and Combinatorial Optimization. {{Algorithms}} and Techniques},
  author = {Guruswami, Venkatesan and Velingker, Ameya and Velusamy, Santhoshini},
  editor = {Jansen, Klaus and Rolim, José D. P. and Williamson, David and Vempala, Santosh S.},
  date = {2017-08},
  series = {{{LIPIcs}}},
  volume = {81},
  pages = {8:1--8:19},
  publisher = {Schloss Dagstuhl — Leibniz-Zentrum für Informatik},
  doi = {10.4230/LIPIcs.APPROX-RANDOM.2017.8},
  eventdate = {2017-08-16/2017-08-18},
  eventtitle = {{{APPROX}} 2017},
  year = {2017}
}

@article{CGSV24,
  title = {Sketching {{Approximability}} of {{All Finite CSPs}}},
  author = {Chou, Chi-Ning and Golovnev, Alexander and Sudan, Madhu and Velusamy, Santhoshini},
  date = {2024-04-12},
  journal = {Journal of the ACM},
  shortjournal = {J. ACM},
  volume = {71},
  number = {2},
  pages = {15:1-15:74},
  publisher = {Association for Computing Machinery},
  doi = {10.1145/3649435},
  note = {Conference version in FOCS 2021},
  year= {2024}
}

@inproceedings{CGV20,
  booktitle = {2020 {{IEEE}} 61st {{Annual Symposium}} on {{Foundations}} of {{Computer Science}}},
  author = {Chou, Chi-Ning and Golovnev, Alexander and Velusamy, Santhoshini},
  date = {2020-11},
  pages = {330--341},
  publisher = {IEEE Computer Society},
  doi = {10.1109/FOCS46700.2020.00039},
  eventdate = {2020-11-16/2020-11-19},
  eventtitle = {{{FOCS}} 2020},
  title = {Optimal {{Streaming Approximations}} for All {{Boolean Max-2CSPs}} and {{Max-}}{\(k\)}{{SAT}}},
  year={2020}
}

@inproceedings{KK19,
  title = {An Optimal Space Lower Bound for Approximating {{MAX-CUT}}},
  booktitle = {Proceedings of the 51st {{Annual ACM SIGACT Symposium}} on {{Theory}} of {{Computing}}},
  author = {Kapralov, Michael and Krachun, Dmitry},
  date = {2019-06},
  pages = {277--288},
  publisher = {Association for Computing Machinery},
  doi = {10.1145/3313276.3316364},
  eventdate = {2019-06-23/2019-06-26},
  eventtitle = {{{STOC}}},
  year = {2019}
}

@inproceedings{KKS15,
  title = {Streaming Lower Bounds for Approximating {{MAX-CUT}}},
  booktitle = {Proceedings of the 26th {{Annual ACM-SIAM Symposium}} on {{Discrete Algorithms}}},
  author = {Kapralov, Michael and Khanna, Sanjeev and Sudan, Madhu},
  date = {2015-01},
  pages = {1263--1282},
  publisher = {{Society for Industrial and Applied Mathematics}},
  doi = {10.1137/1.9781611973730.84},
  eventdate = {2015-01-04/2015-01-06},
  eventtitle = {{{SODA}}},
  year ={2015}
}

@inproceedings{KPSY23,
  title = {Characterizing the {{Multi-Pass Streaming Complexity}} for {{Solving Boolean CSPs Exactly}}},
  booktitle = {14th {{Innovations}} in {{Theoretical Computer Science Conference}}},
  author = {Kol, Gillat and Paramonov, Dmitry and Saxena, Raghuvansh R. and Yu, Huacheng},
  editor = {Tauman Kalai, Yael},
  date = {2023},
  series = {{{LIPIcs}}},
  volume = {251},
  pages = {80:1--80:15},
  publisher = {Schloss Dagstuhl – Leibniz-Zentrum für Informatik},
  doi = {10.4230/LIPIcs.ITCS.2023.80},
  eventdate = {2023-01-10/2023-01-13},
  eventtitle = {{{ITCS}} 2023}
}

@inproceedings{SSSV25,
  title = {Streaming {{Algorithms}} via {{Local Algorithms}} for {{Maximum Directed Cut}}},
  booktitle = {Proceedings of the 2025 {{Annual ACM-SIAM Symposium}} on {{Discrete Algorithms}}},
  author = {Saxena, Raghuvansh and Singer, Noah G. and Sudan, Madhu and Velusamy, Santhoshini},
  date = {2025},
  pages = {3392--3408},
  publisher = {{Society for Industrial and Applied Mathematics}},
  doi = {10.1137/1.9781611978322.111},
  eventdate = {2025-01-12/2025-01-15},
  eventtitle = {{{SODA}} 2025}
}

@inproceedings{FMW25,
  title = {Multi-{{Pass Streaming Lower Bounds}} for {{Approximating Max-Cut}}},
  booktitle = {Proceedings of the 66th {{IEEE Symposium}} on {{Foundations}} of {{Computer Science}}},
  author = {Fei, Yumou and Minzer, Dor and Wang, Shuo},
  date = {2025-12},
  publisher = {IEEE Computer Society},
  eventdate = {2025-12-14/2025-12-17},
  eventtitle = {{{FOCS}} 2025}
}

@inproceedings{FMW26,
  title = {A {{Dichotomy Theorem}} for {{Multi-Pass Streaming CSPs}}},
  booktitle = {Proceedings of the 58th {{ACM Symposium}} on {{Theory}} of {{Computing}}},
  author = {Fei, Yumou and Minzer, Dor and Wang, Shuo},
  date = {2026},
  publisher = {Association for Computing Machinery},
  eventdate = {2026-06-22/2026-06-26},
  eventtitle = {{{STOC}} 2026},
  year = {2026}
}

@misc{STV25,
  title = {Sketching Approximations and {{LP}} Approximations for Finite {{CSPs}} Are Related},
  author = {Singer, Noah G. and Tulsiani, Madhur and Velusamy, Santhoshini},
  date = {2025-09-23},
  eprint = {2509.17926},
  eprinttype = {arXiv},
  eprintclass = {cs},
}

@inproceedings{ABFS26,
  title = {Half-{{Approximating Maximum Dicut}} in the {{Streaming Setting}}},
  booktitle = {Proceedings of the 58th {{ACM Symposium}} on {{Theory}} of {{Computing}}},
  author = {Azarmehr, Amir and Behnezhad, Soheil and Ferrante, Shane and Saneian, Mohammad},
  date = {2026},
  publisher = {Association for Computing Machinery},
  eventdate = {2026-06-22/2026-06-26},
  eventtitle = {{{STOC}} 2026}
}

@misc{Vel25,
  title = {Near-Optimal Streaming Approximation for {{Max-DICUT}} in Sublinear Space Using Two Passes},
  author = {Velusamy, Santhoshini},
  date = {2025-12-22},
  eprint = {2512.19521},
  eprinttype = {arXiv},
  eprintclass = {cs},

}

@misc{FMW26-tight,
  title = {Near-{{Optimal Space Lower Bounds}} for {{Streaming CSPs}}},
  author = {Fei, Yumou and Minzer, Dor and Wang, Shuo},
  date = {2026-04-01},
  eprint = {2604.01400},
  eprinttype = {arXiv},
  eprintclass = {cs},
}

@misc{ABF26,
  title = {Single-{{Pass Streaming CSPs}} via {{Two-Tier Sampling}}},
  author = {Azarmehr, Amir and Behnezhad, Soheil and Ferrante, Shane},
  date = {2026-04-02},
  eprint = {2604.01575},
  eprinttype = {arXiv},
  eprintclass = {cs},
}

@misc{STV26,
  title = {Optimal Single-Pass Streaming Lower Bounds for Approximating CSPs},
  author = {Singer, Noah and Tulsiani, Madhur and  Velusamy, Santhoshini},
  date = {2026-04-09},
  eprint = {2604.08731},
  eprinttype = {arXiv},
  eprintclass = {cs},
}

@article{Index1,
title = {Lower bounds for one-way probabilistic communication complexity and their application to space complexity},
journal = {Theoretical Computer Science},
volume = {157},
number = {2},
pages = {139-159},
year = {1996},
issn = {0304-3975},
doi = {10.1016/0304-3975(95)00157-3},
url = {https://www.sciencedirect.com/science/article/pii/0304397595001573},
author = {Farid Ablayev}
}

@article{HellN08,
  author       = {Pavol Hell and
                  Jaroslav Nesetril},
  title        = {Colouring, constraint satisfaction, and complexity},
  journal      = {Comput. Sci. Rev.},
  volume       = {2},
  number       = {3},
  pages        = {143--163},
  year         = {2008},
  url          = {https://doi.org/10.1016/j.cosrev.2008.10.003},
  doi          = {10.1016/J.COSREV.2008.10.003},
  bibsource    = {dblp computer science bibliography, https://dblp.org}
}

@article{HellN92,
  author       = {Pavol Hell and
                  Jaroslav Nesetril},
  title        = {The core of a graph},
  journal      = {Discret. Math.},
  volume       = {109},
  number       = {1-3},
  pages        = {117--126},
  year         = {1992},
  url          = {https://doi.org/10.1016/0012-365X(92)90282-K},
  doi          = {10.1016/0012-365X(92)90282-K},
  bibsource    = {dblp computer science bibliography, https://dblp.org}
}

@inproceedings{ChandraM77,
  author       = {Ashok K. Chandra and
                  Philip M. Merlin},
  editor       = {John E. Hopcroft and
                  Emily P. Friedman and
                  Michael A. Harrison},
  title        = {Optimal Implementation of Conjunctive Queries in Relational Data Bases},
  booktitle    = {Proceedings of the 9th Annual {ACM} Symposium on Theory of Computing,
                  May 4-6, 1977, Boulder, Colorado, {USA}},
  pages        = {77--90},
  publisher    = {{ACM}},
  year         = {1977},
  url          = {https://doi.org/10.1145/800105.803397},
  doi          = {10.1145/800105.803397},
  bibsource    = {dblp computer science bibliography, https://dblp.org}
}

@incollection{BartoKW17,
  author       = {Libor Barto and
                  Andrei A. Krokhin and
                  Ross Willard},
  editor       = {Andrei A. Krokhin and
                  Stanislav Zivn{\'{y}}},
  title        = {Polymorphisms, and How to Use Them},
  booktitle    = {The Constraint Satisfaction Problem: Complexity and Approximability},
  series       = {Dagstuhl Follow-Ups},
  volume       = {7},
  pages        = {1--44},
  publisher    = {Schloss Dagstuhl - Leibniz-Zentrum f{\"{u}}r Informatik},
  year         = {2017},
  url          = {https://doi.org/10.4230/DFU.Vol7.15301.1},
  doi          = {10.4230/DFU.VOL7.15301.1},
  bibsource    = {dblp computer science bibliography, https://dblp.org}
}

@inproceedings{KolaitisV08,
  author       = {Phokion G. Kolaitis and
                  Moshe Y. Vardi},
  editor       = {Nadia Creignou and
                  Phokion G. Kolaitis and
                  Heribert Vollmer},
  title        = {A Logical Approach to Constraint Satisfaction},
  booktitle    = {Complexity of Constraints - An Overview of Current Research Themes
                  [Result of a Dagstuhl Seminar]},
  series       = {Lecture Notes in Computer Science},
  volume       = {5250},
  pages        = {125--155},
  publisher    = {Springer},
  year         = {2008},
  url          = {https://doi.org/10.1007/978-3-540-92800-3_6},
  doi          = {10.1007/978-3-540-92800-3_6},
  bibsource    = {dblp computer science bibliography, https://dblp.org}
}

@inproceedings{Bulatov17,
  author       = {Andrei A. Bulatov},
  editor       = {Chris Umans},
  title        = {A Dichotomy Theorem for Nonuniform CSPs},
  booktitle    = {58th {IEEE} Annual Symposium on Foundations of Computer Science, {FOCS}
                  2017, Berkeley, CA, USA, October 15-17, 2017},
  pages        = {319--330},
  publisher    = {{IEEE} Computer Society},
  year         = {2017},
  url          = {https://doi.org/10.1109/FOCS.2017.37},
  doi          = {10.1109/FOCS.2017.37},
  bibsource    = {dblp computer science bibliography, https://dblp.org}
}

@inproceedings{BrakensiekGP26,
  author       = {Joshua Brakensiek and
                  Venkatesan Guruswami and
                  Aaron Putterman},
  editor       = {Nicolas Beldiceanu},
  title        = {Classification of Non-Redundancy of Boolean Predicates of Arity 4},
  booktitle    = {32nd International Conference on Principles and Practice of Constraint
                  Programming, {CP} 2026, Lisbon, Portugal, July 20-23, 2026},
  series       = {LIPIcs},
  volume       = {379},
  pages        = {8:1--8:24},
  publisher    = {Schloss Dagstuhl - Leibniz-Zentrum f{\"{u}}r Informatik},
  year         = {2026},
  url          = {https://doi.org/10.4230/LIPIcs.CP.2026.8},
  doi          = {10.4230/LIPICS.CP.2026.8},
  bibsource    = {dblp computer science bibliography, https://dblp.org}
}

@article{brakensiek2026shrinking,
  author       = {Joshua Brakensiek and
                  Venkatesan Guruswami and
                  Bart M. P. Jansen and
                  Victor Lagerkvist and
                  Magnus Wahlstr{\"{o}}m},
  title        = {Super-linear Lower Bounds for {CSP} Non-Redundancy via Shrinking Instances},
  journal      = {CoRR},
  volume       = {abs/2605.19055},
  year         = {2026},
  url          = {https://doi.org/10.48550/arXiv.2605.19055},
  doi          = {10.48550/ARXIV.2605.19055},
  eprinttype   = {arXiv},
  eprint       = {2605.19055},
  bibsource    = {dblp computer science bibliography, https://dblp.org}
}

@article{SV26b,
  author       = {Amatya Sharma and
                  Santhoshini Velusamy},
  title        = {Non-Redundancy of Low-Arity Symmetric Boolean CSPs},
  journal      = {CoRR},
  volume       = {abs/2605.14007},
  year         = {2026},
  url          = {https://doi.org/10.48550/arXiv.2605.14007},
  doi          = {10.48550/ARXIV.2605.14007},
  eprinttype   = {arXiv},
  eprint       = {2605.14007},
  bibsource    = {dblp computer science bibliography, https://dblp.org}
}

@inproceedings{LGE26,
  author       = {Victor Lagerkvist and
                  Johanna Groven and
                  Leif Eriksson},
  editor       = {Sven Koenig and
                  Chad Jenkins and
                  Matthew E. Taylor},
  title        = {Towards Single Exponential Time for Temporal and Spatial Reasoning:
                  {A} Study via Redundancy and Dynamic Programming},
  booktitle    = {Fortieth {AAAI} Conference on Artificial Intelligence, Thirty-Eighth
                  Conference on Innovative Applications of Artificial Intelligence,
                  Sixteenth Symposium on Educational Advances in Artificial Intelligence,
                  {AAAI} 2026, Singapore, January 20-27, 2026},
  pages        = {14287--14294},
  publisher    = {{AAAI} Press},
  year         = {2026},
  url          = {https://doi.org/10.1609/aaai.v40i17.38443},
  doi          = {10.1609/AAAI.V40I17.38443},
  bibsource    = {dblp computer science bibliography, https://dblp.org}
}

@article{BZ20,
  title = {Sparsification of {{Binary CSPs}}},
  author = {Butti, Silvia and {\v Z}ivn{\'y}, Stanislav},
  year = {2020},
  month = jan,
  journal = {SIAM Journal on Discrete Mathematics},
  volume = {34},
  number = {1},
  pages = {825--842},
  publisher = {{Society for Industrial and Applied Mathematics}},
  issn = {0895-4801},
  doi = {10.1137/19M1242446}
}
